\documentclass[oneside]{amsart}
\usepackage{color,caption}
\usepackage{latexsym}
\usepackage{mathrsfs,amsmath,amsthm,amsopn,amsfonts,amssymb,epsfig,stmaryrd,mathtools,dsfont}
\usepackage{graphicx}
\usepackage[latin1]{inputenc}
\usepackage[english]{babel}
\usepackage[T1]{fontenc}
\usepackage{eucal,enumitem,cite}
\usepackage{txfonts}
\begingroup
  \makeatletter
  \@for\theoremstyle:=definition,remark,plain\do{%
    \expandafter\g@addto@macro\csname th@\theoremstyle\endcsname{%
      \addtolength\thm@preskip{.5\baselineskip plus .2\baselineskip minus .2\baselineskip}
      \addtolength\thm@postskip{.5\baselineskip plus .2\baselineskip minus .2\baselineskip}
    }%
  }
\endgroup
\usepackage[capitalise,noabbrev]{cleveref}
\DeclareSymbolFont{largesymbols}{OMX}{zplm}{m}{n}

\usepackage{tikz}

\usepackage{ytableau}

\makeatletter
\newcommand{\superimpose}[2]{%
  {\ooalign{$#1\@firstoftwo#2$\cr\hfil$#1\@secondoftwo#2$\hfil\cr}}}
  \makeatother
  \newcommand{\PlusX}{\mathpalette\superimpose{{\times}{+}}     }
  \newcommand{\PlusXb}{\mathpalette\superimpose{{\PlusX}{\bullet}}     }

\newcommand{\sFill}[1]{
		\begin{tikzpicture}[scale=.125, remember picture]
	\begin{scope}[overlay]
		\draw[clip] (0,.5) circle [radius=1];
	\end{scope}
	\begin{scope}[overlay]
		\draw[clip] (0,.5) circle [radius=1];
		\node[scale=.8] at (0,.5) {${#1}$};
	\end{scope}
	\end{tikzpicture}
}

\newcommand{\yF}[1]{\none[{\sFill{#1}}]}

\newcommand{\superY}[1]{\begin{array}{c} \scalebox{1}{\begin{ytableau}#1 \end{ytableau} } \end{array}}

\ytableausetup{smalltableaux}

\let\l\Bigl
\let\r\Bigr

\theoremstyle{definition}
\newtheorem{example}{Example}[section]

\theoremstyle{plain}
\newtheorem{theorem}[example]{Theorem}
\newtheorem{lemma}[example]{Lemma}
\newtheorem{proposition}[example]{Proposition}
\newtheorem{corollary}[example]{Corollary}

\theoremstyle{remark}
\newtheorem{remark}[example]{Remark}
\newtheorem{note}[example]{Note}

\let\d\partial

\numberwithin{equation}{section}

\def\be{\begin{equation}}
\def\ee{\end{equation}}

\newcommand{\exterior}{\textstyle\bigwedge}

\def\ta{\theta}
\def\La{\Lambda}
\def\la{\lambda}
\def\Om{\Omega}
\def\dd{\mathrm{d}} 
\def\ii{\mathfrak{i}} 

\def\m{{\mathrm m}}

\begin{document}

\title{Bernstein operators and super-Schur functions: combinatorial aspects}  

\author{L.~Alarie-V\'ezina}

\address[Ludovic Alarie-V\'ezina]{
D\'epartement de Physique, de G\'enie Physique et d'Optique \\
Universit\'e Laval \\ 
Qu\'ebec, Canada, G1V~0A6.
}

\email{ludovic.alarie-vezina.1@ulaval.ca}

\author{O.~Blondeau-Fournier}

\address[Olivier Blondeau-Fournier]{
D\'epartement de Physique, de G\'enie Physique et d'Optique \\
Universit\'e Laval \\ 
Qu\'ebec, Canada, G1V~0A6.
}

\email{olivier.b-fournier.1@ulaval.ca}

\author{L.~Lapointe}

\address[Luc Lapointe]{
Instituto de Matem\'atica y F\'{\i}sica \\ 
Universidad de Talca \\
 2 norte 685, Talca, Chile.
}

\email{lapointe@inst-mat.utalca.cl}

\author{P.~Mathieu}

\address[Pierre Mathieu]{
D\'epartement de Physique, de G\'enie Physique et d'Optique \\
Universit\'e Laval \\ 
Qu\'ebec, Canada, G1V~0A6.
}

\email{pmathieu@phy.ulaval.ca}

\thanks{\today}
\maketitle


\begin{abstract}{
    The Bernstein vertex operators, which can be used to build
    recursively the Schur functions, are extended to superspace.
    Four families of super vertex operators are defined, corresponding to the four natural families of Schur functions in superspace.  Combinatorial proofs
    that the super Bernstein vertex operators
    indeed build   the Schur functions in superspace recursively are provided.
    We briefly mention a possible realization, in terms of symmetric functions
    in superspace,
    of the super-KP hierarchy, where the tau-function naturally expands in one of the super-Schur bases. 
}
\end{abstract}

\section{Introduction}

The ubiquitous Schur function $s_\la(z)$ in infinitely many variables $z_1, z_2, \ldots$ and associated to a partition $\la$, can be obtained from the Macdonald polynomial $P_\la(z;q,t)$  by letting $q=t$ \cite{MacSym95}.  
Equivalently, one can add an intermediate step by setting $q=t^\alpha$ and taking the limit $t\rightarrow 1$, thereby generating the Jack symmetric polynomial {$P_\la^{(\alpha)}(z)$}; letting $\alpha=1$  then gives back the Schur function
  $s_\la(z)$.  An important feature of the Schur functions is the
  simplicity of their expansion in the monomial basis: all expansion coefficients are non-negative integers that have a combinatorial interpretation as fillings of certain diagrams.

\medskip
Following this procedure to obtain a supersymmetric extension of the  Schur functions turns out to be somewhat more intricate. 
The Macdonald superpolynomial $P_\La(z,\eta;q,t)$, which is parametrized by a superpartition $\La$ (defined in Section \ref{SSandGF}), has been introduced in \cite{BDLM12, BDLM12_2}.  
 In addition to be polynomials in the $z$ variables, it also depends on infinitely many {odd} (or anticommuting) variables $\eta_1, \eta_2, \ldots$, {(often) referred to as Grassmann variables. }
We should note that, as is the case for the Macdonald polynomials,
the two free parameters $q$ and $t$  are related to the parameters of the (supersymmetric) Ruijsenaars-Schneider model \cite{BDM15} for which the Macdonald (super)polynomials are eigenfunctions.   
When $q=t^\alpha$ and $t\to 1$, $P_\La(z,\eta;q,t)$ reduces to the Jack superpolynomial $P_\La^{(\alpha)}(z,\eta)$, which is an eigenfunction of the supersymmetric Calogero-Sutherland model \cite{DLM_CMS}.

\medskip

Somewhat surprisingly, the specialization of the Jack superpolynomials at $\alpha=1$ does not have appealing combinatorial properties such as a positive expansion coefficients in the monomial basis. 
The combinatorially interesting basis whose elements deserve to be called Schur superfunctions, denoted $s_\La(z,\eta)$, has been uncovered in \cite{BDLM12}.  
They are obtained from the Macdonald superpolynomial $P_\La(z,\eta;q,t)$ by setting $q=t=0$ (which is not a particular specialization of the corresponding Jack superpolynomial):
\be\label{defS} 
s_\La(z,\eta)=P_\La(z,\eta;0,0).
\ee 
 Not only their monomial expansion coefficients are  positive integers, but they also 
satisfy a generalized version of the Macdonald positivity conjecture (see for instance \cite{BDLM12, BLM15}).

Because the Macdonald superpolynomials do not satisfy anymore the inversion symmetry $q,t \leftrightarrow q^{-1},t^{-1}$,  it is equally natural to consider the 
limiting case 
\be 
\bar s_\La(z,\eta)= 
{P_\La(z,\eta;\infty,\infty)}
\ee
which are also called Schur superfunctions. 

\medskip

The description of the family
of super-Schurs does not end here. In the usual case, the Schur functions are self-dual with respect to the Hall scalar product. Now, this scalar product has a natural extension to superspace with respect to which  $s_\La$ and $\bar s_\La$ are almost dual:  denoting by $s_\La^*(z,\eta)$ and $\bar s^*_\La(z,\eta)$
the dual bases to $s_\La$ and $\bar s_\La$ respectively, that is, such that 
\be 
\langle s_\La^{\phantom *},s_\Om^*\rangle = \delta_{\La\Om} \, ,
\qquad  \qquad  \langle \bar s_\La^{\phantom *},\bar s_\Om^*\rangle = \delta_{\La\Om}
\ee
we have that $s_\La^*(z,\eta)$ and $\bar s^*_\La(z,\eta)$ are obtained from $\bar s_\La(z,\eta)$ and $s_\La(z,\eta)$ through the action of a natural automorphism (see \eqref{sesbeomsign1} for more details).

Thus, disregarding the supposedly uninteresting $\alpha=1$ version of Jack superpolynomials, we are left with four versions of the super-Schurs (although, as we just mentioned, we can consider them as two pairs related through the action of an automorphism).

\medskip

Even though they are expected to be fundamental objects, the super-Schurs do not have a simple definition akin to the determinantal definition of the usual Schurs. The various ways of defining for instance the super-Schurs $s_\Lambda(z,\eta)$ known so far are:
\begin{enumerate}
\item from the Macdonald superpolynomials via \eqref{defS};
\item from the Pieri rules \cite{BM1, JL17}; 
\item combinatorially, in terms of fillings of a super-diagram \cite{BM1, JL17}. 
\end{enumerate}
Here we present another construction by generalizing the Bernstein vertex operator $B(z)$.  To be more precise, the Schur functions can be defined from their modes, $B(z)= \sum_{k\in \mathbb Z}B_k z^k$ according to
\be
s_\la = B_{k_1}\, B_{k_2}\ldots B_{k_n} \cdot { 1} \qquad \text{if}\qquad \la = (k_1, k_2, \ldots, k_n).
\ee
In this article, we display the superspace counterpart of that construction for the four versions of super-Schur functions. 
The motivation for this work, apart from its intrinsic interest, is spelled out in the conclusion, where we sketch our plan to use these results to formulate a supersymmetric extension of the KP hierarchy,  to be worked out in a sequel article.

\medskip

The present article is organized as follows. In Section~\ref{SpBo} we recall some basic results concerning symmetric functions which are then used to construct explicitly the Bernstein operators.  A brief review of the concepts of superpartitions and superpolynomials are presented in Section~\ref{SuperP}, emphasizing the generalization of the multiplicative classical bases and the Schur functions. As already indicated, there are four types of Schur superpolynomials and consequently, there are four versions of super Bernstein operators.  
We begin Section~\ref{SBI} by introducing two of them, while the other two families are presented at the end of the section after a brief overview of the properties of a second automorphism.
  In Section~\ref{PerpofBCas}, we study the adjoints of the super Bernstein operators and their relation with the negative modes that act on super-Schur functions by removing columns.  
In Section~\ref{ToSKPzzz}, we conclude with a presentation of the general plan, to be addressed in a sequel article, regarding the connection with the super-KP hierarchy.    
Three appendices complete this work.  Since the construction of Schur (super)polynomials is based on Pieri rules, we review, in Appendix~\ref{PieriApp}, the four rules pertaining the the super-case.  More precisely, we review three of them and work out the required fourth one, which turns out to be new.  Appendix~\ref{ExtraRels} is devoted to detailing basic but non standard manipulations of the classical bases.  Finally, a proof of a proposition which parallels closely one already included in the main text has been relegated to Appendix~\ref{Sproof_typeI}.

\section*{Acknowledgements}
The authors acknowledge the support of the Centre de Recherches Math\'ematiques at the Universit\'e de Montr\'eal where part of this work was done.     
This work was partially supported by FONDECYT (Fondo Nacional de Desarrollo Cient\'{\i}fico y Tecnol\'ogico de Chile) regular grant \#{1170924} (L.~L.) and by NSERC.

\section{Schur polynomials and Bernstein operators}
\label{SpBo}
\subsection{Review of some results in symmetric function theory}
Let $A_{n, \mathbb Q}=\mathbb{Q}[z_1, \ldots, z_n]^{\mathfrak S_n}$ be the ring of symmetric polynomials in the variables $z_1, \ldots, z_n$.  
The symmetric group $\mathfrak S_n$ acts by permuting the variables. Set
\be x_i=\frac{1}{i} ( z_1^i + z_2^i + \ldots ), \qquad i=1,2,\ldots
 \ee  
 (the usual $i$-th power-sum, up to a constant).   
 In the following, we shall be interested in the limit when the number of variables is infinite ($n\rightarrow \infty$). We may thus simply denote by  $A_{\mathbb Q}=\mathbb{Q}[x_1, x_2, \ldots]$ the ring of symmetric functions. 
  Recall that bases are parametrized by partitions $\la=(\la_1, \la_2, \ldots)$, and that for a multiplicative basis we write $x_\la = x_{\la_1}x_{\la_2}\cdots$.
  Denote by $e_\la$, $h_\la$, $m_\la$ and $s_\la$ respectively the elementary, homogeneous, monomials and Schur symmetric functions.  Since any basis is independent over  $A_{\mathbb Q}$, we have naturally $e_i=e_i(x_1, x_2, \ldots)$,  $h_i=h_i(x_1, x_2, \ldots)$, etc.   We will (often) drop the explicit dependence on the variables $x_i$.
  
The partition function is given by:
\be\label{partfZ}
\mathrm Z(x;\bar x) = \sum_{\la} \frac{ x_\la \bar x_\la }{\mathsf z_\la} 
= \sum_{\la}h_\la(x) m_\la(\bar{x}) 
=\sum_{\la}s_\la(x) s_\la(\bar{x}) 
=\exp \Bigl(\sum_{k>0} k x_k \bar{x}_k \Bigr) 
\ee
where $\bar{x}=\bar{x}_1, \bar{x}_2, \ldots$ represents another set of variables.  The weight $\mathsf z_\la$ is given by\footnote{The weight  $\mathsf z_\la$ differs from the usual one $z_\lambda$ by a factor ${\prod_i \la_i}$ squared.}
\be
\mathsf z_\la =  \frac{|\mathrm{Aut}(\la)|}{\prod_i \la_i}
\ee
where $|\mathrm{Aut}(\la)|$ is the order of the group of automorphism that fixes $\la$:
\be\label{defeqAutla}
|\mathrm{Aut}(\la)| = \prod_{i\geq 1} n_i ! \qquad \text{if} \qquad \la=(1^{n_1}, 2^{n_2}, \ldots ).
\ee
Two particular specializations of the partition function $\mathrm Z(x;\bar x)$ give the generating series of the homogeneous and the elementary symmetric functions: 
\begin{align}
H(z) &= \sum_{k=0}^\infty z^k h_k = 
\mathrm Z\Bigl(x; z, \frac{z^2}{2}, \frac{z^3}{3}, \ldots\Bigr)  {=  \exp\l( \sum_{k>0}z^k x_k \r) }
\\
E(z)& = \sum_{k=0}^\infty z^k e_k=  
\mathrm Z\Bigl(x; z, - \frac{z^2}{2}, \frac{z^3}{3}, \ldots\Bigr) {= \exp \l(- \sum_{k>0}(-z)^k x_k \r)}
\label{genseriesEz}
\end{align}
with $h_0=e_0=1$. 
It is immediate that $H(z)E(-z)=1$.  Note  that  
$H(z)$ is also the generating series of the Schur functions indexed by partitions with only one part.
Expression \eqref{partfZ} is equivalent to the formulation of a scalar product $\langle * , * \rangle:A_{\mathbb Q} \times A_{\mathbb Q}\rightarrow \mathbb Q$,  given by
\be\label{sp1}
\langle x_\la, x_\mu \rangle = \delta_{\la\mu}
\mathsf z_\la.
\ee
 With respect to this scalar product, we have the following dualities
\be
\langle h_\la, m_\mu\rangle =  \langle s_\la, s_\mu\rangle = \delta_{\la \mu},
\ee
which imply in particular that the Schur functions form an orthonormal basis.

\medskip

For an element $f\in A_{\mathbb Q}$, we denote the adjoint map $\perp\,:\,$ $f\mapsto f^\perp$ defined to satisfy
\be\label{defadj}
\langle f \ast, \ast \rangle =  \langle \ast, f^\perp \ast \rangle
\ee
for arbitrary elements $\ast$.   Under the adjoint map, the element $x_k$ is {mapped} to
\be\label{xperpno1}
x_k^\perp = \frac{1}{k} \frac{\partial}{\partial x_k}= \frac{1}{k} \partial_{x_k}
\ee
so that an element $f$ which depends on $x$ transforms as
\be
f^\perp(x_1, x_2, x_3, \ldots) = f(\partial_{x_1}, \tfrac12 \partial_{x_2}, \tfrac13 \partial_{x_3}, \ldots) .
\ee
For instance, the adjoint of the generating series {$E(z)$ }given in \eqref{genseriesEz} {is} easily found to be
\be  \begin{split}
  E^\perp (z) &= \sum_{k=0}^\infty z^k e_k^\perp = \exp \Bigl(- \sum_{k>0}\frac{(-z)^k}{k}\partial_{x_k} \Bigr).\label{gen_serE_perp}
  \end{split}
\ee

\subsection{Bernstein operators}

The Bernstein vertex operator $B(z)$ is defined by
\be\label{Bvertexxdelx}
B(z) = \sum_{k\in\mathbb Z}B_kz^k  =  \exp\l( \sum_{k>0}z^k x_k \r) \,  \exp \l(- \sum_{k>0} \frac{z^{-k}}{k} \partial_{x_k} \r)
\ee
and the Bernstein operators (or modes) $B_k, k\in \mathbb Z$ can be obtained from the (residue) relation
\be
B_k = \oint \frac{\dd z}{2\pi \ii} z^{-k-1} B(z).
\ee
By noting that $B(z) = H(z) E^{\perp}(-z^{-1})$, it is straightforward to obtain the explicit expression for the modes of the vertex operator,
\be\label{Bern_m}
B_k = \sum_{r\geq 0} (-1)^r h_{k+r}^{\phantom \perp} \circ e_r^\perp
\ee
where the symbol $\circ$ denotes the composition (or regular multiplication, following the notation of \cite{MacSym95}).
Observe that the Bernstein operator $B(z)$ acting on the identity, which is denoted by 
$B(z) \cdot1$,  is simply the generating function of Schur functions indexed by partitions with a single part.  In fact, Bernstein operators can be viewed as creation operators for the spectrum of Schur functions.
\begin{theorem} \cite{MacSym95} \label{TheoBBB}
Given a partition $\la= (\la_1, \la_2, \ldots, \la_n)$, we have that
\be\label{SchurasBm}
s_\la = 
B_{\la_1} \ldots B_{\la_n} \cdot 1
\ee
In other words, the Schur function $s_\la$ is obtained by the action of an ordered sequence of Bernstein operators acting on the identity.
\end{theorem}

  A direct proof of this result is sketched in  \cite[p. 96]{MacSym95}, and a combinatorial proof is included as a special case of the proof presented in Section \ref{Sproof_typeIe}.

\medskip

When computing \eqref{SchurasBm}, one can either use \eqref{Bern_m},
which depends on the Pieri rules, or \eqref{Bvertexxdelx}, which relies on multiplications and derivations in the variables $x_i$.  We shall be more interested in the first method.  Recall the Pieri rule:
\be
h_r s_\la = \sum_\mu s_\mu,
\ee
where the sum is over partitions such that $\mu/ \la$ is a horizontal $r$-strip and the \emph{inverse} Pieri rule 
\be 
e_r^\perp s_\la = \sum_\mu s_\mu
\ee
where the sum is over partitions such that  $\la/\mu$ is a vertical $r$-strip.  Note that the formula for the inverse Pieri rule follows from the definition of the adjoint \eqref{defadj} and the fact that Schur functions form an orthonormal basis.  It acts by {\it removing} a vertical strip.

\medskip

To illustrate this computation,  consider the following simple example.  Suppose we want to  check Theorem~\ref{TheoBBB}  in the case  $\la=(3,1)$. We write
\be\nonumber
s_{(3,1)}=B_3 B_1\cdot 1
\ee
with $B_1\cdot 1 = h_1= s_{(1)}$.  Then, for $B_3s_{(1)}$,  we see from  \eqref{Bern_m} that only the terms $r=0$ and $r=1$ can contribute since it is not possible to remove a 2 vertical strip from $\la=(1)$.  Thus
\be\nonumber
B_3B_1 \cdot 1=h_3s_{(1)}-h_4= (s_{(4)}+s_{(3,1)}) - s_{(4)}
\ee
which is the desired result.

\medskip

In brief, the understanding of the Pieri rule (and its inverse) may be used to obtain creation operators for the Schur functions. This is the approach we shall consider for the super-Schur functions.

\section{Superpolynomials and generating functions}
\label{SuperP}
\subsection{Superpartitions and multiplicative bases}\label{SSandGF}
We now consider the extension of the  ring of symmetric polynomials  $A_{n,\mathbb Q}$  with the exterior algebra of Grassmann (odd) variables $\eta_1, \ldots, \eta_n$
\be
\left(\mathbb{Q}[z_1, \ldots, z_n] \otimes \exterior [\eta_1, \ldots, \eta_n] \right)^{\mathfrak S_n}, \qquad  \quad \eta_i \eta_j= -\eta_j \eta_i.
\ee
  The action of the symmetric group $\mathfrak S_n$ is extended to the variables $\eta_i$ by permuting them in the same way as  the $z_i$ variables (which it is often referred to as the diagonal action).  
Superpolynomials that are invariant with respect to $\mathfrak S_n$ are called symmetric superpolynomials, and in the limit $(n\rightarrow \infty)$ of infinitely many variables, we may simply consider the ring algebra
\be
\mathcal A_{\mathbb Q} = A_{\mathbb Q}  \otimes \exterior [\ta_1, \ta_2, \ldots]
\ee
of symmetric superfunctions in variables $x_i$ and $\ta_i$ where the latter are the odd component of the power-sum superfunctions:
\be \theta_i=\eta_1 z_1^{i-1}+ \eta_2 z_2^{i-1}+ \ldots
\ee
(Note that the variable $\ta_i$ corresponds to the fermionic power-sum $\tilde p_{i-1}$ that can be found for instance in \cite{DLM_basis}).

\medskip 

Elements of $\mathcal A_{\mathbb Q}$ are indexed by superpartitions $\La=(\La^{\rm a}; \La^{\rm s})=(\La_1, \ldots, \La_{\rm m}; \La_{\rm m+1}, \ldots, \La_N)$ where $\La^{\rm a}$ is a strictly decreasing partition (that may include a 0), $\La^{\rm s}$ is a non increasing  partition, and 
$\mathrm{m} = \mathrm{m}(\La)$ will denote the fermionic degree of $\La$ (i.e.~the number of parts in $\La^{\rm a}$).   For later use, we denote by $\La^+$ the partition obtained by removing the semi-colon and reordering the parts in nonincreasing order.  
We let  $\La^* = \La^+$ and $\La^\circledast= (\La^{\rm a}+1; \La^{\rm s})^+$ where $\La^{\rm a}+1$ means that we add one to every part of $\La^{\rm a}$.  
The parts  of $\La^{\rm a}$ 
are said to be fermionic while those of $\La^{\rm s}$ 
are called bosonic.  Obviously, if $\La^*_i\neq \La^\circledast_i$  then $\La_i^*$ is fermionic whereas if $\La^*_i= \La^\circledast_i$ then that part is bosonic.    

\medskip

The diagrammatical representation of a superpartition is obtained as follows:  we first represent the diagram of $\La^\circledast$ and then transform the boxes of the skew diagram $\La^\circledast/\La^*$ in circles.  
 Here is an example
 \be\label{exsuperpart1}
 \La=(8,6,3,2,0;5,3)  \quad \rightarrow \quad 
\scalebox{1.1}{$\superY{
\,&\,&\,&\,&\,&\,&\,&\,& \yF{} \\
\,&\,&\,&\,&\,&\,&\yF{}\\
\,&\,&\,&\,&\\
\,&\,&\,& \yF{}\\
\,&\,&\\
\,&\,&\yF{}\\
\yF{}
}$}
\ee
The conjugate of a superpartition $\Lambda$, denoted $\La'$ is obtained by interchanging the rows and columns in the diagrammatical representation.  For instance, using the example \eqref{exsuperpart1}, we have $(8,6,3,2,0;5,3)' = (6,5,3,1,0;6,3,2,1)$.

\medskip

Given a superpartition $\La$, we write $\mathcal{X}_\La$ for the power-sum (multiplicative) basis over $\mathcal A_{\mathbb Q}$, and defined by
\be\label{XXLA}
\mathcal{X}_\La = \ta_{\La_1+1}\cdots \ta_{\La_{\mathrm{m}}+1} \, 
x_{\La_{\mathrm m+1}} \cdots x_{\La_N}.
\ee
Let $\bar x=\bar x_1, \bar x_2, \ldots$ and $\bar \ta = \bar \ta_1, \bar \ta_2, \ldots$ be two other sets of indeterminates (even and  odd respectively), and let $\bar{\mathcal{X}}_\La$ be defined as \eqref{XXLA} but with $\ta_j$ replaced by $\bar \ta_j$ and $x_j$ replaced by $\bar x_j$ for all $j$'s.  
The partition function in superspace is
\be\label{superpartZ}
\mathcal Z(x,\ta;\bar x, \bar \ta) =\exp\l( \sum_{k>0} k x_k \bar{x}_k \r) \, \exp \l(\sum_{k>0} \ta_k \bar{\ta}_k \r) = 
\sum_\La   \frac{\mathcal{X}_\La \, \bar{\mathcal{X}}_\La}{\mathsf z_\La} ,
\ee
  where the weight is now
  \be
   \mathsf z_\La = (-1)^{\binom{{\mathrm m}}{2}} 
    \frac{|\mathrm{Aut}(\La^{\rm s})|}{\prod_i \La^{\rm s}_i}. 
  \ee
  (Recall that $\mathrm{m}=\mathrm{m}(\Lambda)$ is the fermionic degree of $\La$ and the expression for $|\mathrm{Aut}(\la)|$ is given in \eqref{defeqAutla}).  
  The scalar product between elements of $\mathcal A_{\mathbb Q}$ is defined by
\be\label{pssuperpptoto}
\langle * , * \rangle: \mathcal A_{\mathbb Q} \times \mathcal A_{\mathbb Q} \rightarrow \mathbb Q\qquad \quad
 \mathcal X_\La\times\mathcal X_\Omega \mapsto \langle \mathcal X_\La, \mathcal X_\Omega \rangle = \delta_{\La \Omega}\mathsf z_\La.
\ee

We denote by $h_\La$ and $e_\La$ respectively the homogeneous and elementary symmetric functions in superspace {which are also  multiplicative bases over $\mathcal A_{\mathbb Q}$ }\cite{DLM_basis}, i.e.,
\be 
f_\La=\tilde f_{\La_1}\cdots \tilde f_{\La_{\m}}f_{\La_{\m+1}}\cdots f_{\La_N}\qquad f=h,e
\ee
Generating functions for the constituents $(h_n,\tilde h_n)$ and $(e_n,\tilde e_n)$ can be used to defined these two  bases in terms of $x_i$ and $\ta_i$.  
These are 
 obtained by specializing each $\bar x_k$ and $\bar\ta_k$ in the partition function \eqref{superpartZ}, i.e., by setting $\bar{x}_k=z^k/k, \bar{\ta}_k=-\eta z^{k-1}$ and (respectively) $\bar{x}_k=-(-z)^k/k, \bar{\ta}_k=-\eta (-z)^{k-1}$:
\begin{align}\label{GFhe}
H(z;\eta)& = \sum_{k\geq 0} z^k ( h_k + \eta \tilde{h}_k) =   \exp\l( \sum_{k>0} z^k \l(x_k + \frac{\eta}{z} \ta_k\r)\r)  \\ \label{GFhee}
E(z;\eta)& = \sum_{k\geq 0} z^k ( e_k + \eta \tilde{e}_k) =   \exp \l(- \sum_{k>0} (-z)^k \l(x_k + \frac{\eta}{z} \ta_k\r)\r) .
\end{align}
We introduce the automorphism $\omega: \mathcal A_{\mathbb Q} \rightarrow \mathcal A_{\mathbb Q} $ defined as
\be\label{superomega}
\omega:\; x_n \mapsto (-1)^{n-1} x_n, \qquad \ta_n \mapsto (-1)^{n-1} \ta_n, \qquad n=1,2, \ldots
\ee
This provides a duality relationship between $H(z;\eta)$ and $E(z;\eta)$:
\be
H(z;\eta) = \omega\bigl( E(z;\eta) \bigr)
\ee
which leads to the relations
\be\label{dualEH}
\omega(e_n) = h_n, \qquad \omega(\tilde{e}_n) = \tilde{h}_n; \qquad \omega(h_n) = e_n, \qquad \omega(\tilde h_n)= \tilde e_n.
\ee

Finally, the adjoint map $\perp\, :\,$ $f\mapsto f^\perp$ for any element $f\in \mathcal A_{\mathbb Q}$ is defined similarly as the adjoint map in \eqref{defadj} but using instead the scalar product \eqref{pssuperpptoto}. Thus, in addition to \eqref{xperpno1}, the variable $\ta_n$ transforms as $\ta_n^\perp=\partial_{\ta_n}$, and for a general superfunction $f$ of the $x_i, \ta_i$ we have
\be\label{fadjinsupers}
f^\perp(x_1, x_2, \ldots ; \ta_1, \ta_2, \ldots)=f( \partial_{x_1}, \tfrac12 \partial_{x_2}, \ldots ;  \partial_{\ta_1}, \partial_{\ta_2}, \ldots).
\ee

For instance, the adjoint of the generating series for the homogeneous functions is
\be
H^\perp(z;\eta) = \sum_{k\geq 0} z^{k} (h_k^\perp - \eta  \tilde{h}^\perp_k) = \exp \l(\sum_{k>0} z^k \l( \frac{1}{k} \partial_{x_k} - \frac{\eta}{z} \partial_{\ta_k} \r)\r) .
\ee
Note in this last expression the relative minus sign coming from the reordering of odd variables in the adjoint map, i.e.~$(\eta \tilde h_k)^\perp = \tilde h_k^\perp \eta= -\eta \tilde h_k^\perp$.

\subsection{Schur superpolynomials}
As mentioned in the introduction, there are four families of superfunctions which can be considered genuine analogues of the usual Schur functions.  We call them Schur superpolynomials (of super-Schur for short) by analogy.   They naturally appear in the decomposition of the partition function in superspace \eqref{superpartZ}:
\be
  \mathcal Z(x,\ta; \bar x, \bar \ta)= \sum_\La s_\La^{\phantom *}(x,\ta) s_\La^*(\bar{x},\bar{\ta}) =  \sum_\La \bar{s}_\La^{\phantom *}(x,\ta) \bar{s}_\La^*(\bar{x},\bar{\ta})
  \ee
  or equivalently, they satisfy the duality relations
\be\label{orthops_ssesbsbe}
\langle  s_\La^{\phantom *}, \,  s_\Omega^* \rangle = \delta_{\La \Omega}, \qquad
 \langle \bar{s}_\La^{\phantom *}, \,  \bar{s}_\Omega^* \rangle = \delta_{\La \Omega}
\ee
w.r.t.~the scalar product \eqref{pssuperpptoto}.  Only two families are truly  independents, namely $s_\La$ and $\bar s_\La$; the remaining ones are obtained through the action of an automorphism (see below).    
Their original definition was presented in \cite{BDLM12_2} as a limiting case of the Macdonald superpolynomials, and a combinatorial definition (by mean of tableaux fillings) was presented as a conjecture in \cite{BM1}.  In the following sections, we shall use another approach to give a precise definition of these super-Schur functions: in terms of creation operators similar to the Bernstein operators.    
We will refer to the different families of super-Schur functions as of type I, I$^*$, II, and II$^*$:
\be\label{les4famss}
\begin{tabular}{| l || c | c | c | c |} 
\hline
Type: & I & I$^*$ &  II  &  II$^*$  \\
\hline\hline
Super-Schur: &  $s_\La^{\phantom *}$   &  $s_\La^*$     &   $\bar{s}^{\phantom *}_\La $ & $ \bar{s}_\La^*$   \\
\hline
\end{tabular}
\ee
and we will denote accordingly the corresponding super Bernstein
operators.%
\footnote{We stress the different meanings {of} the upper script $*$ when used for $\La$ or for $s_\La$.  {In} particular, we distinguish $s^*_\La$ and $s_{\La^*}$: the former is the dual of the super-Schur $s_\La$, while the latter is an ordinary Schur function (since $\La^*$ is an ordinary partition).}

\medskip

It can be shown \cite{BDLM12_2} that the  automorphism $\omega$ in \eqref{superomega} yields the following relations between the different super-Schurs:
\be\label{sesbeomsign1}
s_\La^* =  \omega( \bar s_{\La'}), \qquad \bar{s}_\La^* =  \omega  (s_{\La'})
\ee
where we recall that $\La'$ stands for the conjugate superpartition.  We then see that the super-Schurs of types I and II are sufficient in order to obtain types I$^*$ and II$^*$.   Equivalently, in the next section, we will consider types I and  I$^*$ as the two independent families, and derive types II and II$^*$ from them.  

 \medskip

 The different super-Schurs indexed by superpartitions with  a single row or a single column, given   in terms of the multiplicative bases, are summarized in Table \ref{tab1}.  
Note that for a regular partition, the four families of super-Schur all correspond to the same Schur function.

\begin{table}[ht]
\caption{Schur superfunctions with one row or one column, for $r\geq 0\, \, (e_0=h_0=1)$.}
\label{tab1}
\begin{center}
\begin{tabular}{|c||c|c| c| c |} 
\hline
$\La$ &$(0;1^r)$&$(r; \;)$ &  $(\; ; 1^r)$ & $(\; ; r)$\\
\hline\hline
$s_\La$&$\tilde e_r$&$\ta_{r+1}$ &$e_r$ & $h_r$ \\
      $s_\La^*$&$\tilde e_0 e_r$&$\tilde h_r$&  $e_r$ & $h_r$ \\
     $\bar s_\La$&$\tilde e_r$&$\tilde h_0 h_r$ &  $e_r$ & $h_r$ \\
      $\bar s_\La^*$&$(-1)^r\ta_{r+1}$&$\tilde h_r$&  $e_r$ & $h_r$ \\     
\hline
\end{tabular}
\end{center}
\end{table}

The above expressions for the various one-row or one-column versions of the different super-Schur entail a fairly large number of Pieri rules.  
The list of all Pieri rules is given in Table \ref{tab2}.  These are paired two-by-two by the duality operation.

\medskip

\begin{table}[ht]
\caption{The different Pieri rules}
\label{tab2}
\begin{center}
\begin{tabular}{||c|c||} 
\hline
Pieri rules & $\omega$-dual rules\\
\hline\hline
   $s_\La\,h_r$ & $\bar s_\La^*\, e_r$ \\
     $ s_\La\, \ta_{r}$ & $\bar s_\La^*\,\ta_r$  \\ 
     \hline
$s_\La\,e_r$ & $\bar s_\La^*\,h_r$  \\
$s_\La\,\tilde e_r$ & $\bar s_\La^*\,\tilde h_r$ \\
\hline
   $s_\La^*\,h_r$&$\bar s_\La\, e_r$ \\
     $ s_\La^*\,\tilde h_r$&$\bar s_\La\,\tilde e_r$\\ \hline  
$s_\La^*\,e_r$&$\bar s_\La\,h_r$\\
$s_\La^*\,\tilde e_0e_r$&$\bar s_\La\,\tilde h_0h_r$\\
\hline
\end{tabular}
\end{center}
\end{table}

Looking at Table \ref{tab2},  note that the first four Pieri rules were first conjectured in \cite{BM1}. These together with the following two were then proved in \cite{JL17}. The remaining two Pieri rules have not been considered before. The penultimate one will be needed below. This rule is thus formulated and proved here.  The relevant results pertaining to the Pieri rules are collected in Appendix \ref{PieriApp}.

\medskip

\section{Super Bernstein operators}
\label{SBI}

In this section, we present an explicit construction of the
 four families of
super-Schur functions (see \eqref{les4famss}) in terms of 
super Bernstein operators.  
This is done by introducing the superspace generalization of the generating series  \eqref{Bvertexxdelx} and showing that their (positive) Laurent modes are creation operators for the super-Schurs.  The explicit expressions for the Laurent modes are mainly motivated by the Pieri rules (which are reviewed in Appendix~\ref{PieriApp}), and hence most of the results follow from combinatorial properties of the super-Schurs.

\medskip

We shall use the generating series that we already introduced, together with additional auxiliary generating series,  to define the super vertex operators.  Let $\partial_{\tilde e_r}$ denote the derivative w.r.t.~the odd symmetric elementary function $\tilde e_r$, which can be written as
\be\label{deriv}
\partial_{\tilde{e}_r} = \frac{\partial}{\partial \tilde{e}_r} = (-1)^r \sum_{s\geq 0} h_s \partial_{\ta_{r+s+1}}, \qquad (h_0=1),
\ee
 (see Appendix \ref{ExtraRels}), and let $\partial_{\tilde e_r}^\perp$ denote its adjoint as defined in  equation \eqref{fadjinsupers}.  Set
 \begin{align}
 \mathsf b(z;\eta) =   \bigl(\exp \eta \partial_{\tilde e_0} \bigr) \, \sum_{r\geq0} (-z)^r \tilde{e}_r, \qquad
 \mathsf c(z;\eta)  = \exp\bigl( \eta \sum_{r\geq 0} (-z)^r \partial_{\tilde e_r}^\perp \bigr)
 \end{align}
and also
\be
 \bar{\mathsf c}(z;\eta) =  \bigl(\exp \eta \partial_{\ta_1} \bigr)  \sum_{r \geq 0} z^r \ta_{r+1}.
\ee
Let $V_\pm(z)$ be the formal generating series of the modes $\beta_{\pm n}$ given by
\be \label{VpmetBeta}
V_\pm(z) = \exp \bigl( - \sum_{n>0} \frac{z^{\pm n}}{\pm n} \beta_{\pm n}\bigr) , \qquad 
\beta_n = \sum_{r>0} \ta_{r+n} \partial_{\ta_r}.
\ee
We are now ready to give the four types of vertex operators in superspace.  They will be denoted as
\be\label{les4famsBCBC}
\begin{tabular}{| l || c | c | c | c |} 
\hline
Type: & I & I$^*$ &  II  &  II$^*$  \\
\hline\hline
Bernstein superoperator: &  $B(z;\eta)$   &  $C(z;\eta)$     &   $\bar{B}(z;\eta)$ & $ \bar{C}(z;\eta)$   \\
\hline
\end{tabular}
\ee
where
\begin{align}
B(z;\eta) & = H(z) \, \mathsf b(z;\eta)   \, E^{\perp}(-z^{-1})  \label{genBB}\\
 C(z;\eta) & = H(z) \, \mathsf c(z;\eta)   \, E^{\perp}(-z^{-1})  \label{SVO_C}  \\ 
 \bar{B}(z;\eta) & = V_+(z) \, H(z;\eta) \, E^{\perp}(-z^{-1}) \label{genbarBBzz} \\ 
\bar{C}(z;\eta) & = H(z) \,   \bar{\mathsf c}(z;\eta) \,   E^{\perp}(-z^{-1}) \, V_-(z) \label{genbarCC}. 
\end{align}
In the following subsections, we will show that the expansion in modes of the generating series \eqref{genBB}-\eqref{genbarCC} define the super-Schur functions.

\subsection{Type-I Bernstein operators} Let $\mathbb Z_2$ denote the set $\{0,1\}$.  Expanding the type I generating series as
\be
B(z;\eta) = \sum_{n\in \mathbb Z, \epsilon  \in \mathbb Z_2} z^n \eta^{1-\epsilon} B_n^{(\epsilon)}  = \sum_{n\in \mathbb Z} z^n \bigl( \eta B_n^{(0)} + B_n^{(1)} \bigr)
\ee
we obtain, by simple manipulations, the following Laurent modes
\be\label{B0vsB1}
B_n^{(1)}= \sum_{r\geq0} (-1)^r \ta_{n+r+1}^{\phantom \perp} \circ e_r^\perp, \qquad \qquad
 B_n^{(0)}  = \partial_{\tilde{e}_0} \circ B_{n}^{(1)},
\ee
defined for all $n\in \mathbb Z$, and where the symbol $\circ$ denotes the composition.  
Equivalently, the {mode} $B_n^{(0)} $ can be given explicitly as
\be\label{Bm0}
B_n^{(0)} = \sum_{r\geq 0}(-1)^r h_{n+r}^{\phantom \perp}\circ e_r^\perp - \sum_{r,s\geq 0} (-1)^{r+s} \ta_{n+r+s+1}^{\phantom \perp}\circ e_r^\perp \circ \partial_{\tilde{e}_{s}}
\ee
(the equivalence between these two expressions for $B^{(0)}_n$ is shown below).
The modes $B_n^{(0)}$ and $B_n^{(1)}$ are called respectively the even and odd super Bernstein operators of type I.  They define an action $B_n^{(\epsilon)}\colon \mathcal A_{\mathbb Q}\rightarrow \mathcal A_{\mathbb Q}$ on the ring of symmetric superpolynomials.  
In the following, we shall mostly focus on the positive modes $(n\geq 0)$.  
The even operator increases the total bosonic degree of a superfunction by $n$ without affecting its fermionic degree, while the odd operator increases its total degree by $n$ and its fermionic degree by one. They satisfy the property:  
\begin{proposition} \label{prop_typeI}
  For a superpartition $\Lambda$, let $\epsilon_i=\epsilon_i(\La)=\La^\circledast_i-\La^*_i\in \mathbb Z_2$, for $i=1,2,\ldots, N$, that is, let $\epsilon_i$  be $0$ for a bosonic part and $1$ for a fermionic part. 
Then, we have
\be\label{BBB1}
s_\La  = B_{\La^*_1}^{(\epsilon_1)} \ldots  B_{\La^*_N}^{(\epsilon_N)} \cdot {1} . 
\ee
In other words, the super-Schur $s_\La$ is obtained by acting  with a string of operators $B_n^{(\epsilon)}$ on the identity. As such,   
it is a supersymmetric analogue of \eqref{SchurasBm}.
\end{proposition}

The proof of Proposition \ref{prop_typeI} is relegated to Appendix \ref{Sproof_typeI}. We illustrate this result with a few examples. 

\begin{example} For the superpartition $(4,1;3,2,2)$, the associated Schur function is given by 
   $$
   s_{(4,1;3,2,2)}=B_4^{(1)}B_3^{(0)}B_2^{(0)}B_2^{(0)}B_1^{(1)}\cdot 1$$
   since $(4,1;3,2,2)^*=(4,3,2,2,1)$, the parts $4$ and $1$ are fermionic (associated to the odd operators) while the parts $3,2,2$ are bosonic (associated to the even operators).  
\end{example}

\begin{example}
We will verify that  $B_3^{(1)} s_{(0;3)} = s_{(3,0;3)}$. 
 In this case, we will not use the inverse Pieri rule of the type $e_r^\perp s_\La$ but use instead  the explicit expression of $e_r^\perp$ in terms of derivatives with respect to the variables $x_n$. It can be shown that the expansion of $s_{(0;3)}$ in the
 powers-sum basis is given by
\[
 s_{(0;3)}=-\ta_4+\ta_1 x_3+\ta_1x_1x_2+\frac16 \ta_1x_1^3,
 \]
For the action of  
 $e_r^\perp$ we use the generating series 
\eqref{gen_serE_perp} to obtain
  \[ e_1^\perp=\d_{x_1},\qquad e_2^\perp=\frac12\d^2_{x_1}-\frac12 \d_{x_2},\qquad e_3^\perp=\frac13 \d_{x_3}-\frac12\d_{x_1}\d_{x_2}+\frac16 \d_{x_1}^3,\qquad \text{etc.}
  \]
which implies that
\[ 
e_3^\perp s_{(0;3)}=0,\qquad  e_2^\perp s_{(0;3)}=0,\qquad  e_1^\perp s_{(0;3)}=\ta_1x_2+\frac12\ta_1x_1^2=s_{(2;0)}+s_{(0;2)}
\]
where it can be checked that indeed $\ta_1x_2+\ta_1x_1^2/2=s_{(2;0)}+s_{(0;2)}$. Hence, we get (since $e_0^\perp=1$)
\[
 B_3^{(1)} s_{(0;3)}=\ta_4  s_{(0;3)}-\ta_5 ( s_{(2;)}+ s_{(0;2)}).
 \]
To complete the computation, we use the Pieri rule of type $\ta_r s_\La$, given in Appendix \ref{PieriApp}, {which produces}
\[\begin{split}
\ta_4 \, s_{(0;3)}&=s_{(6,0;)}+s_{(5,1;)}+s_{(4,2;)}+s_{(5,0;1)}+s_{(4,0;2)}+s_{(3,0;3)}\\
\ta_5 \, s_{(0;2)}&=s_{(6,0;)}+s_{(5,1;)}+s_{(5,0;1)}+s_{(4,0;2)}\\
\ta_5 \, s_{(2;)}&=s_{(4,2;)}\end{split}\]
Summing all the contributions,  we find the expected result.  
\end{example}

\medskip

\begin{note} \label{NoteBn0zzz}
We end this section by showing how to obtain \eqref{Bm0} from   \eqref{B0vsB1} for the mode $B_n^{(0)}$.   This follows from 
applying directly the derivative $\partial_{\tilde{e}_0}$ in \eqref{B0vsB1} on $B_n^{(1)}$: using 
the expression \eqref{deriv}, we obtain
\be\label{db0}
\partial_{\tilde{e}_0}\circ B_n^{(1)}
= \sum_{r\geq 0} (-1)^r [ {h}_{n+r}^{\phantom \perp} \circ e^\perp_r -
 { \ta}_{n+r+1}^{\phantom \perp} \circ  \partial_{\tilde{e}_0} \circ e_r^\perp].
\ee 
To prove the equivalence between the above expression  and \eqref{Bm0}, {we need to show the following formula}
\be \label{theRtemp1}
  \sum_{r,s\geq 0} (-1)^{r+s} \ta_{n+r+s+1}^{\phantom \perp}\circ e_r^\perp \circ \partial_{\tilde{e}_{s}}= \sum_{r\geq 0} (-1)^r { \ta}_{n+r+1}^{\phantom \perp} \circ  \partial_{\tilde{e}_0} \circ e_r^\perp.
 \ee
 {Using again \eqref{deriv}, the left-hand side of \eqref{theRtemp1} is}
 \be \label{theRtemp2}
 \sum_{r,s,t\geq 0} (-1)^r \ta_{n+r+s+1} \circ e_{r}^\perp \circ h_t\circ \partial_{\ta_{s+t+1}}.
\ee 
Then, we make use of the exchange relation 
 \[ e_{r}^\perp \circ h_t=   h_t \circ e_{r}^\perp+h_{t-1} \circ e_{r-1}^\perp
 \]
 and the fact that $e_r^\perp \circ \partial_{\ta} =  \partial_{\ta} \circ e_r^\perp$,
to write \eqref{theRtemp2} as
\[ \begin{split}
\sum_{r,s,t\geq 0} (-1)^r \ta_{n+r+s+1}^{\phantom \perp} \circ [h_t \circ e_{r}^\perp+h_{t-1} \circ e_{r-1}^\perp
]\circ \partial_{\ta_{s+t+1}}=
\sum_{r,s,t\geq 0} (-1)^r h_t\, [\ta_{n+r+s+1} \partial_{\ta_{s+t+1}} -  \ta_{n+r+s+2} \partial_{\ta_{s+t+2}}] \circ e_r^\perp,
\end{split} \]
where we have shifted the indices in the second term and used $h_{-1}=e_{-1}^\perp=0$.
The sum over $s$ can be performed explicitly (it is a telescopic series), and only the term $s=0$ in the first sum  remains.  {Hence, the expression  \eqref{theRtemp2} results into}
\[
\sum_{r,t\geq 0 } (-1)^r h_t \ta_{n+r+1}\partial_{\ta_{t+1}} \circ e_r^\perp =
\sum_{r\geq 0 } (-1)^r \ta_{n+r+1}\l(\sum_{t\geq 0 }h_t\,\partial_{\ta_{t+1}} \r)\circ e_r^\perp= \sum_{r \geq 0} (-1)^r \ta_{n+r+1}^{\phantom \perp} \circ \partial_{\tilde{e}_0} \circ e_r^\perp
\]
where we used again \eqref{deriv}, this time with $s=0$. This completes the 
desired verification.
\end{note}

\subsection{Type-I$^*$ Bernstein operators}
We now present the creation operators for the Schur superfunctions of type I$^*$,
namely for the $s^*_\La$'s.   {Referring to \eqref{les4famsBCBC} and \eqref{SVO_C}, and expanding this time the generating series as
\be
C(z;\eta) = 
\sum_{n\in \mathbb Z, \epsilon  \in \mathbb Z_2} z^n \eta^{\epsilon} C_n^{(\epsilon)} 
\ee
we obtain the Laurent modes
\be\label{eqdefCm0}
C_n^{(0)}  = \sum_{r\geq 0} (-1)^r h_{n+r}^{\phantom \perp} \circ e_r^\perp
\ee
and
\be  
\label{eqdefCm1}
C_n^{(1)}=\partial_{\tilde{e}_0}^\perp  \circ C_n^{(0)} =  \sum_{r,s,t \geq 0}  (-1)^r h_{n+r-t}^{\phantom \perp} \circ \ta_{s+t+1}^{\phantom \perp} \circ h_s^\perp \circ e_r^\perp
\ee
(defined for all $n\in \mathbb Z$).  This is seen as follow.  We have
\be
C(z;\eta)= H(z) \, \bigl( 1 + \eta  \sum_{s\geq 0} (-1)^s \partial_{\tilde e_s}^\perp \, z^s \bigr) \, E^\perp(-z^{-1})
\ee
and hence the term which is independent of the Grassmman variable $\eta$ is the usual Bernstein operator (denoted $C_n^{(0)}$, as in \eqref{eqdefCm0}).     Then the coefficient of $\eta$, using again formula \eqref{deriv}, is
\be
\partial_\eta C(z;\eta) = \sum_{k,r,s\geq 0} z^{k+s} \ta_{r+s+1}^{\phantom \perp} h_k^{\phantom \perp} \circ h_r^\perp \circ E^\perp(-z^{-1}) .
\ee
If we change the summation indices by $k+s=i\geq 0, s=j\geq0$, we obtain
\be
\partial_\eta C(z;\eta) = \sum_{i,j,r\geq 0} z^i \ta_{j+r+1} h_{i-j} \circ h_r^\perp \circ E^\perp(-z^{-1}),
\ee
and then followed by $j+r=p\geq 0, j=q\geq 0$:
\be
\partial_\eta C(z;\eta)  = \sum_{i,p,q\geq 0} z^i \ta_{p+1} h_{i-q} \circ h^\perp_{p-q} \circ E^\perp(-z^{-1}).
\ee
Performing finally the sum over $q$ using the first identity in \eqref{hhperp}, we have the desired expression
\be
\partial_\eta C(z;\eta)= \sum_{i,p} z^i \ta_{p+1} h^\perp_p \circ h_i \circ E^\perp(-z^{-1}) = \partial_{\tilde e_0}^\perp \circ H(z) \circ E^\perp(-z^{-1}) = \sum_{n\in \mathbb Z} z^n \, \bigl (  \partial_{\tilde e_0}^\perp \circ  C_n^{(0)} \bigr).  
\ee
The explicit formula for \eqref{eqdefCm1} can easily be obtained using manipulations similar to those found in Note~\ref{NoteBn0zzz}.  We thus omit the details of the computation.

\medskip

{
The operators $C_n^{(\epsilon)} \colon \mathcal A_{\mathbb Q} \rightarrow \mathcal{A}_{\mathbb Q}$, for $\epsilon=0,1$, are respectively called the even and odd Bernstein operators of type I$^*$.  In particular, they have the following action.
}

\medskip

\begin{proposition} \label{prop_typeIe}
{Let $\La$ be a superpartition with $\La^*=(\La_1^*, \ldots, \La_N^*)$, and let $\epsilon_i=\epsilon_i(\La)$ be defined as in Proposition \ref{prop_typeI}.  We have
\be\label{CCC1}
s_\La^*  = C_{\La^*_1}^{(\epsilon_1)} \ldots  C_{\La^*_N}^{(\epsilon_N)}\cdot {1} .
\ee
}
\end{proposition}

The proof of this statement will be presented in the following subsection.  

\medskip

\begin{remark}\label{RemCn1alter1}
An alternative expression for the odd Bernstein operator \eqref{eqdefCm1} is given by
\be
C_n^{(1)}=
\sum_{r,s \geq 0} (-1)^r \Bigl( \tilde{h}_{n+r+s}^{\phantom \perp} - \sum_{t=0}^{s-1} \ta_{t+1}^{\phantom \perp} h_{n+r+s-t}^{\phantom \perp}\Bigr)  \circ h^\perp_s\circ e^\perp_r .
\ee
\end{remark}

\begin{example} Consider the super-Schur $s^*_{(2;1)}$.  From Proposition \ref{prop_typeIe}, we have  $s^*_{(2;1)}=C_2^{(1)} C_1^{(0)}\cdot 1$. Since $C_1^{(0)}\cdot 1=s_{(1)}$, we then use the formula presented in Remark \ref{RemCn1alter1}, to compute the remaining action of $C_2^{(1)}$, and note that only the cases with $r=0$ and $r=1$ are non-zero.  This gives
  \[
C_2^{(1)}s_{(1)} = \sum_{s=0}^1 \Bigl(\tilde{h}_{2+s} -\sum_{t=0}^{s-1}\ta_{t+1}h_{2+s-t}\Bigr) h_s^\perp s_{(1)} - \tilde{h}_3 = \tilde{h}_2 s_{(1)}- \ta_1 h_3
\]
which is indeed equal to $s^*_{(2;1)}= \tilde{h}_2h_1 - \tilde{h}_0 h_3$.
\end{example}

\subsection{Proof of Proposition \ref{prop_typeIe}} 
\label{Sproof_typeIe}
In this section, we present a combinatorial proof that 
the Bernstein operators of type I$^*$ satisfy \eqref{CCC1}.   The proof of Proposition \ref{prop_typeI} being very similar, is relegated to Appendix~\ref{Sproof_typeI}. 
 
 \medskip

First, observe that if there are only bosonic parts then Proposition~\ref{prop_typeIe} is trivially verified because the expression
\be
C_{k_1}^{(0)}\ldots C_{k_n}^{(0)} \cdot 1
\ee
gives the usual Schur function $s_\la$ with $\la=(k_1, \ldots, k_n)$ given that the operator $C_k^{(0)}$ is the usual Bernstein operator $B_k$.  

\medskip

It is also straightforward  to verify the proposition in the case where there is  a single fermionic part,
\be
C_k^{(1)}\cdot 1 = \partial_{\tilde{e}_0}^\perp \circ C_k^{(0)} \cdot 1
=\partial_{\tilde{e}_0}^\perp \circ h_k 
 = \sum_{r\geq 0} \ta_{r+1}^{\phantom \perp} h_r^\perp \, h_k^{\phantom \perp} 
 = \sum_{r\geq 0} \ta_{r+1} h_{k-r} 
 = \tilde{h}_k = s^*_{(k;)}.
\ee

The proof of the proposition boils down to proving that, when $n$ is large enough, the action of $C_n^{(\epsilon)}$ on $s^*_\La$ creates a single $s_\Om^*$ where the diagram of $\Om$ is obtained by the addition of a row of length $n$ on the top of the diagram of $\La$.  We first consider the case $\epsilon=0$, that is, we show that if $n \geq \Lambda_1^{\circledast}$ then
\be\label{lemtypeIea}
C_n^{(0)} s^*_\La = s^*_\Om \qquad \text{with} \qquad \Om^*=(n,\La^*)\quad\text{and}\quad \Om^\circledast=(n,\La^\circledast).
\ee

\begin{proof}
Let $\La$ be a superpartition, and let assume that $n\geq \La^\circledast_1$.    
We show that the action of $C_n^{(0)}$ on $s_\La^*$ only produces one Schur function, i.e.
\be\label{cmo_onsetexp1}
C_n^{(0)} s_\La^* = \sum_{r\geq 0}(-1)^r h_{n+r}^{\phantom \perp} \circ e_r^\perp \, s^*_\La =  s^*_{(\La_1^{\rm a}, \ldots, \La_{\rm m}^{\rm a} ;\, n, \La^{\rm s}_1, \ldots)}
\ee
using the Pieri rules.  For that, we will associate a diagram to each super-Schur
functions appearing in $h_{n+r}^{\phantom \perp} \circ e_r^\perp \, s^*_\La$. We will then show that when summing over $r$, all diagrams cancel out two-by-two except an extra one corresponding to $s^*_{(\La_1^{\rm a}, \ldots, \La_{\rm m}^{\rm a} ;\, n, \La^{\rm s}_1, \ldots)}$. To prove that this is indeed the case, we will introduce
a sign-reversing  involution, to be denoted $\mathcal{J}$,  that  pairs each diagram  {-- beside the desired one --} with another of the same shape  but with an  opposite sign.   

\medskip

We begin by  
introducing a pictorial characterization of each diagram appearing in \eqref{cmo_onsetexp1}.  Given a diagram $\La$, the set of {\it removable} boxes of $\La$ forms a $\ell(\La^*)$-vertical strip (where $\ell(\La^*)$ is the length of $\La^*$) and corresponds to all  boxes located at the rightmost position
of each row of $\La^*$.  
They are indicated  by the symbol $\circ$ in the following
 example:
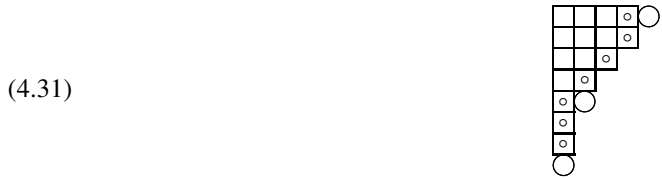
\be \label{ex429}
\scalebox{1.1}{$\superY{
	\,&\,&\,&\circ& \yF{}\\
	\,&\,&\,&\circ \\
	\,&\,&\circ \\
	\,&\circ \\
	\circ &\yF{}\\
	\circ\\
	\circ\\
	\yF{}
	}$}
\ee
The action of $e_r^\perp$ on $s_\La^*$ produces Schur functions (of type I$^*$) given by every diagram that can be obtained from the diagram $\La$ by removing an $r$-vertical strip using the (inverse) Pieri rule of type I.  Thus the set of removable boxes of $\La$ corresponds to the  
boxes that can be removed when  $e_r^\perp$ acts.  
In our pictorial representation, following the action of $e_r^\perp$, instead of removing $r$ boxes from the original diagram $\La$, we leave them in place but mark them with a $\times$.     
By definition of an $r$-vertical strip, only the removable boxes of $\La$ (those marked with a $\circ$ in \eqref{ex429}) can be marked with a $\times$.  
Here is an example
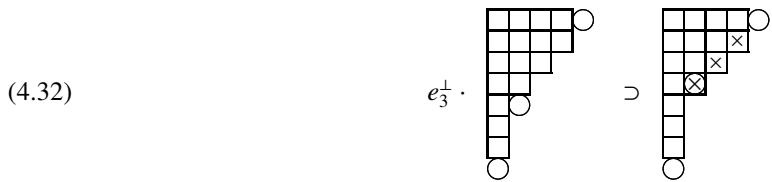
\be
e_3^\perp \cdot
\scalebox{1.1}{$\superY{
\,&\,&\,&\,&\yF{}\\
\,&\,&\,&\, \\
\,&\,&\, \\
\,&\, \\
\,&\yF{}\\
\,\\
\,\\
\yF{}
}$}
\supset
\scalebox{1.1}{$\superY{
\,&\,&\,&\,&\yF{}\\
\,&\,&\,&\times \\
\,&\,&\times \\
\,&\sFill{\times} \\
\,\\
\,\\
\,\\
\yF{}
}$}
\ee
Note that in the above diagram the box with a circle and a $\times$ means that the box originally in this position was removed and the circle in the row just below took its position. 
With this notation, every diagram obtained from  $e_r^\perp s^*_\La$
  is uniquely  characterized.

\medskip

Next, the action of $h_{r}$ on $s_\La^*$ produces Schur functions associated to every diagrams that can be obtained from $\La$ by adding an $r$-horizontal strip  using the Pieri rule of type I$^*$.  Those boxes added to $\La$ are marked with the symbol $+$.
For instance, we have
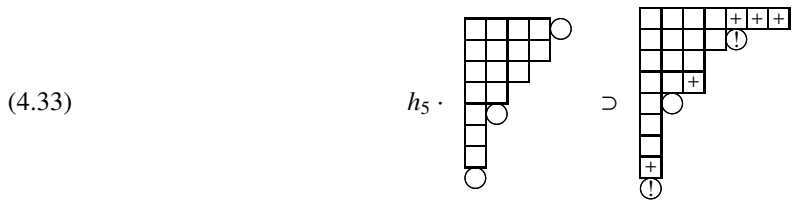
\be
h_{5} \cdot
\scalebox{1.1}{$\superY{
\,&\,&\,&\,&\yF{}\\
\,&\,&\,&\, \\
\,&\,&\, \\
\,&\, \\
\,&\yF{}\\
\,\\
\,\\
\yF{}
}$}
\supset
\scalebox{1.1}{$\superY{
\,&\,&\,&\,&+ &+&+\\
\,&\,&\,&\,&\yF{!} \\
\,&\,&\, \\
\,&\,&+ \\
\,&\yF{}\\
\,\\
\,\\
+\\
\yF{!}
}$}
\ee
Recall that, whenever a box is added to a row which ends with a circle, the circle moves to the next row below (see the above example).   A circle that was moved from its original position in the diagram $\La$ due to the addition of a box will be marked with the symbol ``$!$''.
With this notation, every diagram produced in the expansion $h_r \, s^*_\La$ is
again uniquely characterized.

\medskip

With the notation introduced so far, we can thus represent every super-Schur function appearing in $h_{n+r}^{\phantom \perp }\circ e_r^\perp \, s^*_\La$, where $r=0,1,2,\ldots$, 
by a suitably decorated diagram.  In addition to the boxes marked with the symbols $\times$ and $+$, there are also boxes which are removed in the first step ($\times$) and then replaced in the second step  ($+$).
We will  identify these boxes with  the symbol $\PlusX$. 
      Note that the set of boxes marked with a $\PlusX$ are part of the original diagram; in particular, {they belong to} the set of removable boxes.  
        {Here is an} example,
\be \label{eq432}
h_{8}^{\phantom \perp}\circ e_3^\perp \cdot
	\scalebox{1.1}{$\superY{
	\,&\,&\,&\,&\yF{}\\
	\,&\,&\,&\, \\
	\,&\,&\, \\
	\,&\, \\
	\,&\yF{}\\
	\,\\
	\,\\
	\yF{}
	}$}
\supset
	\scalebox{1.1}{$\superY{
	\,&\,&\,&\,&+&+&+&+&+\\
	\,&\,&\,&\sFill{\times} \\
	\,&\,&\PlusX \\
	\,&\PlusX \\
	\,& \yF{}\\
	\,\\
	\,\\
	+ \\
	\yF{!}
	}$}
\ee
Boxes that are not decorated are simply referred to as empty boxes.  

\medskip

We now introduce the notion of a {\it transmutable} box. It is either an empty box, in which case it will be marked with a $\bullet$, or a box with a $\PlusX$, in which case it will be marked with a $\PlusXb$.  
Reading the diagram from top right  to bottom left, a transmutable box is the first box that satisfies the following conditions:
\begin{itemize}
\item[1.] it is an empty box or a $\PlusX$box; 
\item[2.] it has either no box or only $+$boxes to its right;
\item[3.] it is not above a $+$box or an empty box;
\item[4.]  there is no unmarked (empty) circle to its right.
\end{itemize}
The first two conditions correspond to the fact that a transmutable box can only be  
a removable box of the original diagram. 
By definition, each diagram appearing in $h_{n+r}^{\phantom \perp}\circ e_r^\perp \, s^*_\La$ can have at most one transmutable box.   In the diagram
given in \eqref{eq432}, the transmutable box is the fourth one in the first row, that is, the one marked with a $\bullet$ in the diagram that follows 
\be\label{ex_first_tbbox1}
\scalebox{1.1}{$\superY{
	\,&\,&\,&\bullet&+&+&+&+&+\\
	\,&\,&\,&\sFill{\times} \\
	\,&\,&\PlusX \\
	\,&\PlusX \\
	\,& \yF{}\\
	\,\\
	\,\\
	+ \\
	\yF{!}
	}$}
\ee

\medskip

The involution $\mathcal{J}$ interchanges the transmutable box marked with a $\bullet$ with that marked with a $\PlusXb$ and leaves the diagram unchanged if there is no transmutable box.
  The conditions that define a transmutable box then guarantee that this operation is well defined and is an involution.
When a box marked with a $\bullet$ is interchanged with one marked with a $\PlusXb$, the number of $\times$ changes by one, which is reflected in \eqref{cmo_onsetexp1} by a change of sign.  For instance,
acting with the involution $\mathcal{J}$ on the diagram appearing in \eqref{ex_first_tbbox1}, we obtain
\be\label{ex_first_tbb2}
\scalebox{1.1}{$\superY{
	\,&\,&\,&\PlusX&+&+&+&+&+\\
	\,&\,&\,&\sFill{\times} \\
	\,&\,&\PlusX \\
	\,&\PlusX \\
	\,& \yF{}\\
	\,\\
	\,\\
	+ \\
	\yF{!}
	}$}
\subset
h_{9}^{\phantom \perp}\circ e_4^\perp \cdot
	\scalebox{1.1}{$\superY{
	\,&\,&\,&\,&\yF{}\\
	\,&\,&\,&\, \\
	\,&\,&\, \\
	\,&\, \\
	\,&\yF{}\\
	\,\\
	\,\\
	\yF{}
	}$}
        \ee
     Applying   $\mathcal{J}$  on the previous diagram then obviously
produces the the diagram appearing in \eqref{ex_first_tbbox1}.  

\medskip

There is a unique diagram starting from $\La$ which contains no transmutable box.  
This is 
the diagram in which no box has been removed and a box ($+$) has been added in
the first $n$ column.  
It thus comes from the $r=0$ sector of $h_{n+r}^{\phantom \perp}\circ e_r^\perp \, s^*_\La$ (since $e_0^\perp=1$) and because $n\geq \La_1^*$, such a diagram always exists and is unique. With every column ending with a $+$box,  the point (3.) above rules out the possibility of a transmutable box. 
This diagram is thus the only one mapped to itself under $\mathcal{J}$ and, as such, the only diagram  not cancelled out.  From our example, this left-over diagram is
\be
h_{5}^{\phantom \perp}\circ e_0^\perp \cdot
	\scalebox{1.1}{$\superY{
	\,&\,&\,&\,&\yF{}\\
	\,&\,&\,&\, \\
	\,&\,&\, \\
	\,&\, \\
	\,&\yF{}\\
	\,\\
	\,\\
	\yF{}
	}$}
\supset
	\scalebox{1.1}{$\superY{
	\,&\,&\,&\,&+\\
	\,&\,&\,& & \yF{!} \\
	\,&\,&\,&+ \\
	\,&\, &+\\
	\,&+ \\
	\,& \yF{!}\\
	\,\\
	+ \\
	\yF{!}
	}$}
\ee
which contains no transmutable box.
This completes the proof of Proposition \ref{prop_typeIe}.   We illustrate again the argument of the proof using the following example 
\be
C_3^{(0)} s^*_{(1;3)}  \quad \leftrightarrow \quad  \sum_{r\geq0} (-1)^r h_{3+r}^{\phantom \perp} \circ e_r^\perp \cdot  \scalebox{1.1}{$\superY{
	\,&\,& \circ\\
	\circ&\yF{} 
	}$}
\ee
The different contributions are
\be\begin{split}
& h_3^{\phantom \perp}\circ e_0^\perp \cdot \scalebox{1.1}{$\superY{
	\,&\,& \circ\\
	\circ&\yF{} 
	}$}
=\\
 &\scalebox{1.1}{$\superY{
	\,&\,&\bullet &+&+&+ \\
	\,&\yF{} 
	}$}
+
\scalebox{1.1}{$\superY{
	\,&\,&\bullet &+&+ \\
	\,&+\\
	\yF{!} 
	}$}
+
\scalebox{1.1}{$\superY{
	\,&\,&\bullet &+&+ \\
	\,&\yF{} \\
	+
	}$}
+
\scalebox{1.1}{$\superY{
	\,&\,&\, &+ \\
	\bullet&+&+\\
	\yF{!} 
	}$}
\\
& +
\scalebox{1.1}{$\superY{
	\,&\,&\bullet &+ \\
	\,&+ \\
	+ & \yF{!}
	}$}
+
\scalebox{1.1}{$\superY{
	\,&\,&\,  \\
	\,&+&+ \\
	+ & \yF{!}
	}$}
\end{split}\ee
and
\be\begin{split}
& (-1)h_4^{\phantom \perp}\circ e_1^\perp\cdot \scalebox{1.1}{$\superY{
	\,&\,& \circ\\
	\circ&\yF{} 
	}$}
	=  
	\\
 &
 -\scalebox{1.1}{$\superY{
	\,&\,&\PlusXb &+&+&+ \\
	\,&\yF{} 
	}$}
-
 \scalebox{1.1}{$\superY{
	\,&\,&\PlusXb &+&+ \\
	\,&+\\
	\yF{!} 
	}$}
-
 \scalebox{1.1}{$\superY{
	\,&\,&\PlusXb &+&+ \\
	\,&\yF{} \\
	+
	}$}
-
\scalebox{1.1}{$\superY{
	\,&\,& &+\\
	\PlusXb &+&+ \\
	\yF{!} 
	}$}
\\
&
-
 \scalebox{1.1}{$\superY{
	\,&\,&\PlusXb &+&+ \\
	\,&+\\
	+&\yF{!}\\  
	}$}	
-
\scalebox{1.1}{$\superY{
	\,&\,&\bullet &+&+ &+&+\\
	\sFill{\times} 
	}$}
-
\scalebox{1.1}{$\superY{
	\,&\,&\bullet &+&+ &+\\
	\PlusX \\
	\yF{!} 
	}$}
-
\scalebox{1.1}{$\superY{
	\,&\,&\bullet &+&+\\
	\PlusX &+ \\
	\yF{!} 
	}$}
\end{split}\ee
and
\be\begin{split}
& h_5^{\phantom \perp}\circ e_2^\perp\cdot \scalebox{1.1}{$\superY{
	\,&\,& \circ\\
	\circ&\yF{} 
	}$}
=\\
 &\scalebox{1.1}{$\superY{
	\,&\,&\PlusXb &+&+&+&+ \\
	\sFill{\times}  
	}$}
+
\scalebox{1.1}{$\superY{
	\,&\,&\PlusXb &+&+ &+\\
	\PlusX\\
	\yF{!} 
	}$}
+
\scalebox{1.1}{$\superY{
	\,&\,&\PlusXb &+&+ \\
	\PlusX & +\\
	\yF{!} 
	}$}
\end{split}\ee
where we have marked each transmutable box.  We see that every diagram, except one, is paired -- via the involution $\mathcal{J}$ -- with an identical diagram but with opposite sign.   
We then have
\be
 \sum_{r\geq0} (-1)^r h_{3+r}^{\phantom \perp} \circ e_r^\perp \cdot  \scalebox{1.1}{$\superY{
	\,&\,& \circ\\
	\circ&\yF{} 
	}$}
	=
	\scalebox{1.1}{$\superY{
	\,&\,&\,  \\
	\,&+&+ \\
	+ & \yF{!}
	}$}
\quad
\leftrightarrow \quad
C_3^{(0)} s_{(1;3)}^* = s^*_{(1;3,3)}
\ee
as desired.

\end{proof}

Observe that the previous proof, specialized to the sector $\m=0$, provides a combinatorial proof of Theorem~\ref{TheoBBB}.

\medskip

Finally, we complete the proof of Proposition \ref{prop_typeIe} by considering the case $\epsilon=1$, that is,  we demonstrate that if $n \geq \Lambda_1^{\circledast}$ then 
\be\label{lemtypeIefermion}
C_n^{(1)}s_\La^* = s_{\Om}^* \qquad \text{with} \qquad \Om^* = (n, \La^*) \quad \text{and}\quad \Om^\circledast=(n+1, \La^\circledast).
\ee
\begin{proof}
Using \eqref{eqdefCm1} and \eqref{lemtypeIea}, we can write
\be
C_n^{(1)}s_\La^* = \partial_{\tilde{e}_0}^\perp \circ C_n^{(0)} s^*_\La =  \partial_{\tilde{e}_0}^\perp \circ s_{\Delta}^*
\ee
where $\Delta$ is such that $\Delta^*=(n,\La^*), \Delta^\circledast= (n,\La^\circledast)$. 
 The result then follows from Corollary \ref{corde0tiladj_1} which says that the action of $\partial_{\tilde{e}_0}^\perp$ on $s_\Delta^*$ will produce a single super-Schur function indexed by $\Om $ (corresponding to $\Delta$ with an extra circle in the first row).      
\end{proof}

\subsection{A second automorphism}  
Before we turn to the description of the super-Schur functions of type II and II$^*$ in terms of super Bernstein operators, we need to introduce a second automorphism between basis elements.  This will define a duality relation between Schur superfunctions, and once combined with the usual duality $\omega$ introduced above,  the construction of type II and type II$^*$ will be straightforward.

\medskip

Let $\rho : \mathcal A_{\mathbb Q} \rightarrow  \mathcal A_{\mathbb Q}$ be the automorphism defined by
\be\label{defautrhozz1234}
\rho  ( x_n) = (-1)^{n-1}x_n, \qquad \rho( \ta_n) = \tilde{e}_{n-1}, \qquad n=1,2, \ldots
\ee
Observe that $\rho$, contrary to $\omega$ defined in \eqref{superomega},  mixes the fermionic parts  since $\tilde{e}_{n}$ depends on $x$ and $\ta$.  With this map, the ring of symmetric superpolynomials has new generators
  \be
  \mathcal A_{\mathbb Q} = \mathbb C[h_1, h_2, \ldots]   \otimes \exterior [\ta_1, \ldots, \ta_n] 
  \ee
given that, defining for every superpartition $\La$,
\be\label{labasedeshcheck}
\check{h}_\La = \ta_{\La_1+1} \cdots \ta_{\La_{\mathrm m }+1} h_{\La_{\mathrm m+1}} \cdots h_{\La_\ell} \, \neq h_\La.
\ee
we have that the $\check{h}_\La$ basis is dual to that of the elementary symmetric function:  
\be
\rho(e_\La) = \check{h}_\La, \qquad \rho(\check{h}_\La) = e_{\La}.
\ee
}
Hence, $\rho$ is an  involution ($\rho^2=1$).  From \cite{JL17},  when acting on a super-Schur function of Type I, it gives
\be\label{ladualiterho123456}
\rho(s_\La) = (-1)^{\binom{{\rm m }}{2}} s_{\La'}.
\ee
The adjoint map $\rho^\perp (\neq \rho) $, from the orthogonality relation \eqref{orthops_ssesbsbe}, gives a duality for type I$^*$ super-Schur, or in other words:
\be
\rho^\perp (s_\La^* ) = (-1)^{\binom{\mathrm m}{2}} s_{\La'}^*
\ee
(the map $\rho^\perp$ is also an involution).  
From the scalar product \eqref{orthops_ssesbsbe}, and recalling relations \eqref{sesbeomsign1}, we have 
\be
\begin{split} 
\delta_{\La\Omega} &= \langle  s_\La^{\phantom *}, s_\Omega^* \rangle\\
&=
(-1)^{\binom{{\rm m }}{2}} \langle \rho \circ \rho( s_\La^{\phantom *}), \omega ( \bar{s}_{\Omega'}) \rangle
\\
&=
\langle  s_{\La'}, \rho^\perp \circ \omega( \bar{s}_{\Omega'}  ) \rangle
\\
&=
 (-1)^{\binom{{\rm m }}{2}} \langle \bar{s}^*_\La, \omega \circ \rho^\perp \circ \omega ( \bar{s}_{\Omega'}) \rangle,
\end{split}
\ee
which implies that
\be
\omega \circ \rho^\perp \circ \omega (\bar{s}_{\Omega}) = (-1)^{\binom{{\rm m }}{2}} \bar{s}_{\Om'} \, , \qquad \quad
\omega \circ \rho \circ \omega (\bar{s}_\La^*) = (-1)^{\binom{\mathrm m}{2}} \bar{s}_{\La'}^*.
\ee
Hence, from the two automorphisms $\omega$ and $\rho$ we have obtained the four duality relations for the super-Schur  functions.
\medskip

\begin{remark}  \label{RemTransfoRhoRhoRho} For an element $f\in \mathcal A_{\mathbb Q}$, the derivative w.r.t.~$x_n$ under the homomorphism $\rho$ is
such that   
\[
\rho( \partial_{x_n} f) \neq (-1)^{n-1} \partial_{x_n} \rho(f)
\]
since otherwise this would be equivalent to considering an action of the form $\rho^\perp \circ x_n \circ \rho^\perp$ which cannot be the case because it is known \cite{JL17} that $\rho^\perp$ is not a homomorphism.   Another way of seing this is from the map $\rho \,:\, \ta_n \mapsto \tilde e_{n-1}$, which now has a non-trivial dependence on the $x_n$'s.  Instead, the derivative $\partial_{x_n}$ transforms as
\be\label{rhocdxncrho}
\rho \circ \partial_{x_n} \circ \rho =   (-1)^{n-1} ( \partial_{x_n} + \beta_{-n} )
\ee
(recall the definition of $\beta_n$ in \eqref{VpmetBeta}).  Taking the adjoint on both sides of equation \eqref{rhocdxncrho}, we see that $x_n$ does not transform as nicely as in \eqref{defautrhozz1234} under the automorphism $\rho^\perp$.  
\end{remark}

\medskip

We now combine  the two dualities $\omega$ and $ \rho$ to form their composition
\be
\varphi=\omega\circ \rho.
\ee
  The inverse of $\varphi$ is simply $\varphi^{-1}=\rho\circ\omega$ (since both $\omega, \rho$ are involution).  
Acting on variables $x_n, \ta_n$, we have explicitly
\be\begin{split}
\varphi&: \; x_n \mapsto x_n, \quad \ta_n \mapsto \tilde{h}_{n-1} \\
\varphi^{-1}&: \; x_n \mapsto x_n, \quad \ta_n \mapsto (-1)^{n-1} \tilde{e}_{n-1}
\end{split}
\ee
which shows in particular that  $\omega$ and $\rho$ do not commute.    
Using \eqref{ladualiterho123456}, we can write
\be
\varphi  ( s_\La^{\phantom *} ) = \bar{s}^*_\La, \qquad  \varphi^{-1}  ( \bar{s}^*_\La ) =  s_\La^{\phantom *} 
\ee
so that $\varphi$ interchanges Types I and  II$^*$.
For any two elements $f,g\in \mathcal A_{\mathbb Q}$, it can be checked that:
\be
\langle f,g\rangle = \langle \varphi f, (\varphi^{-1})^\perp g \rangle,
\ee
and thus the map $\varphi^\perp$ relates Type I$^*$ and Type II super-Schurs, i.e.
\be
\varphi^\perp 
( \bar{s}_\La^{\phantom *})  = {s}^*_\La, \qquad 
(\varphi^\perp)^{-1} ( {s}^*_\La) = \bar{s}_\La^{\phantom *}.
\ee

\subsection{Type-II and Type-II$^*$ Bernstein operators}
 We now show how to obtain the super Bernstein operators for the super-Schur functions
$\bar s_\La$ and $\bar{s}^*_\La$.   {This construction is formal in the sense that it only requires
  the Bernstein operators previously introduced and the map $\varphi$.}

\medskip

Looking at Table \eqref{les4famsBCBC}, we expand the generating series $\bar{B}(z;\eta)$ and $\bar{C}(z;\eta)$ respectively in \eqref{genbarBBzz} and in \eqref{genbarCC} as
\be\label{ExpansionbarBC1}
\bar{B}(z;\eta) =  \sum_{n\in \mathbb Z, \epsilon  \in \mathbb Z_2} z^n \eta^{\epsilon} \bar B_n^{(\epsilon)}, \qquad \qquad 
\bar{C}(z;\eta) = \sum_{n\in \mathbb Z, \epsilon  \in \mathbb Z_2} z^n \eta^{1-\epsilon} \bar C_n^{(\epsilon)}
\ee
The Laurent modes  $ \bar B_n^{(\epsilon)}$ (resp.~$\bar C_n^{(\epsilon)}$) for $\epsilon=0$ and $\epsilon=1$ are called the even and odd Bernstein operators of type II (resp.~of type II$^*$).  They satisfy the following relations.

\begin{lemma} \label{LemmbarBC} For $n \in \mathbb Z$, 
we have
\be
 \bar B_n^{(\epsilon)} = (\varphi^\perp)^{-1} \circ C_n^{(\epsilon)} \circ \varphi^\perp,
 \qquad 
  \bar C_n^{(\epsilon)}= \varphi \circ B_n^{(\epsilon)} \circ \varphi^{-1}
\ee
\end{lemma}
\begin{proof} Straightforward. This follows from using the transformation rules between bases, as in Remark \ref{RemTransfoRhoRhoRho}.    \end{proof}

{
In particular, since we have obtained that the set of operators $B_n^{(\epsilon)}$ (resp.~$C_n^{(\epsilon)}$) for $n\geq 0$ generates all super-Schur functions of type I (resp.~of type I$^*$), we have the following.  
}

\begin{corollary}{Let $\La$ be a superpartition with $\La^*=(\La_1^*, \ldots, \La_N^*)$, and let $\epsilon_i=\epsilon_i(\La)$ be defined as in Proposition \ref{prop_typeI}.  Then, we have 
\be
\bar{s}_\La  = \bar{B}_{\La^*_1}^{(\epsilon_1)} \ldots \bar{B}_{\La^*_N}^{(\epsilon_N)}\cdot {1}, \qquad \qquad
\bar{s}^*_\La  = \bar{C}_{\La^*_1}^{(\epsilon_1)} \ldots \bar{C}_{\La^*_N}^{(\epsilon_N)} \cdot {1} .
\ee
}
\end{corollary}
\begin{proof}
The proof of these formulas follows directly from the corresponding formulas for the  Bernstein operators of types I and type I$^*$.  
{We proceed by induction.  Consider the case of $\bar s_\La$.  For one part, we obtain:
\be
\bar{B}_k^{(\epsilon)} \cdot 1 =  (\varphi^\perp)^{-1} \circ C_k^{(\epsilon)} \cdot 1=  (\varphi^\perp)^{-1} \circ
\begin{cases} h_k=s^*_{(\; ; k)} \\ \tilde h_k=s^*_{(k ; \, )}\end{cases}
= 
\begin{cases} \bar{s}_{(\; ; k)}, \qquad (\epsilon=0) \\ \bar{s}_{(k ; \, )}, \qquad (\epsilon=1) \end{cases}.
\ee
Then we need to verify that (for instance), for $n\geq \La_1^\circledast$, we have
\be
\bar{B}_{n}^{(1)} \bar{s}_\La = \bar{s}_\Om, 
\ee}
with $\Om^*=(n,\La^*), \Om^\circledast=(n+1, \La^\circledast)$.  
We use again the expression  from Lemma \ref{LemmbarBC} to write
\be
\bar{B}_{n}^{(1)} \bar{s}_\La =  (\varphi^\perp)^{-1} \circ C_n^{(1)} \circ \varphi^{\perp}  \bar{s}_\La 
=  (\varphi^\perp)^{-1} \circ C_n^{(1)} s^*_\La
=  (\varphi^\perp)^{-1} \circ s^*_\Om 
 = \bar{s}_\Om
\ee
which is the desired result.    
\end{proof}
We stress that this construction is somewhat formal in the sense that  we do not give explicitly the operators in \eqref{ExpansionbarBC1}  in terms of the classical bases and their adjoints.  Note that the explicit expressions can be obtained from the generating series.  For example, we have 
\be
\bar{B}_n^{(0)}= \sum_{r,s\geq0} (-1)^r h_s^{\vphantom\perp}(-\tilde{\beta}_+) \circ h_{n+r-s}^{\phantom\perp}\circ e_r^\perp, \qquad
\bar{B}_n^{(1)}= \sum_{r,s\geq0} (-1)^r h_s^{\vphantom\perp}(-\tilde{\beta}_+) \circ \tilde{h}_{n+r-s}^{\phantom\perp}\circ e_r^\perp,
\ee
where $h_s^{\vphantom\perp}(-\tilde{\beta}_+)$ is the function $h_s(x)$ where each $x_n$ is replaced by the operator $-\beta_n/n$.

\section{Adjoint of the super Bernstein operators and negative modes}
\label{PerpofBCas}

In this section, we consider the action of the adjoint of the super Bernstein operators.  We shall focus on Types I and I$^*$.  So far, we have only considered the positive Laurent modes of the super Bernstein operators.   
We will now obtain the action of new operators that remove columns on the super-Schur functions.    
This will provide a generalization of the standard relation
\be
B^\perp(z) = \omega B(-z^{-1}) \omega,
\ee
which also gives a direct method to
show the action of the negative modes, namely
\be
(-1)^{|\la|} B_{-\la_n'} \ldots B_{-\la_1} \cdot s_\la = 1
\ee
(where $\la'$ denotes the transpose of the partition $\la$).  

\medskip

Let us  introduce the generating series $K(z;\eta)$ and $L(z;\eta)$ given by the expressions
\be
K(z;\eta) = V_+(z)^{-1} \, H(z) \, \mathsf k(z;\eta) \, E^{\perp}(-z^{-1})\ee
where
\be
 \mathsf k(z;\eta)= \bigl( \sum_{r\geq 0} z^{-r} \partial_{\ta_{r+1}} \bigr) \,  \exp (-\eta \ta_1 )
\ee
and
\be
L(z;\eta) = H(z) \, E^\perp(-z^{-1};\eta) \, V_-(z)^{-1}.
\ee
Taking the adjoint of the super Bernstein operators $B(z;\eta)$ and $C(z;\eta)$, these generating series appear in the following way. 
\begin{lemma}\label{LemmBCperpKLrhorho}
We have
\be
B^\perp(z;\eta) = \rho^\perp K(-z^{-1};\eta) \rho^\perp, \qquad C^{\perp}(z;\eta) = \rho L(-z^{-1};\eta) \rho.
\ee
\end{lemma}
\begin{proof} Almost immediate from Lemma \ref{LemmbarBC}.  
\end{proof}

When expanding in Laurent modes in the following way
\be
K(z;\eta)= \sum_{n\in \mathbb Z}z^n( \eta K_n^{(0)} + K_n^{(1)}) , \qquad \quad
L(z;\eta)= \sum_{n\in \mathbb Z}z^n(  L_n^{(0)} + \eta L_n^{(1)})
\ee
the modes have a very simple action on the super-Schurs.    
\begin{corollary}
Let $\La$ be a superpartition, and let $\La'$ be its transpose.  Let $\epsilon'_i=\epsilon_i(\La')$ be defined as in Proposition \ref{prop_typeI}.  Then, we have
\be\label{Res1CorrAdj1}
1 = (-1)^{|\La|} L_{-(\La')^*_N}^{(\epsilon'_N)} \ldots L_{-(\La')^*_1}^{(\epsilon'_1)} \cdot s_\La
\ee
and
\be\label{Res2CorrAdj2}
1 = (-1)^{|\La|} K_{-(\La')^*_N}^{(\epsilon'_N)} \ldots K_{-(\La')^*_1}^{(\epsilon'_1)} \cdot s^*_\La.
\ee
\end{corollary}
\begin{proof}  Consider the action of $C^\perp(z;\eta)$ on the super-Schur $s_\La$ with $\La_1^*=n$.     
From the orthogonality relation \eqref{orthops_ssesbsbe}, we have that ${C^{(\epsilon)\perp}_n} s_\La= s_\Om$ with $\Om^*=(\La_2^*, \La_3^*, \ldots)$, where $\epsilon=0$ if $\La_1$ is bosonic and $\epsilon=1$ if $\La_1$ is fermionic.  
But, from Lemma~\ref{LemmBCperpKLrhorho}, we have that  the action of $C^{(\epsilon) \perp}_n$ is related to the action of $(-1)^nL_{-n}^{(\epsilon)}$ by conjugation with the homomorphism $\rho$ which acts by the duality relation \eqref{ladualiterho123456} on the super-Schur $s_\La$.  Thus $(-1)^nL_{-n}^{(\epsilon)}$ must remove the first row of $\La'$ and  \eqref{Res1CorrAdj1} follows  (and similarly for \eqref{Res2CorrAdj2}).
\end{proof}

\section{Outlook: superpolynomials realization of the SKP hierarchy}  
\label{ToSKPzzz}
In this article, we have presented four families of Bernstein operators in superspace,  each family generating a given Schur function in superspace.  This is summarized in Table \eqref{les4famsBCBC}.  
Apart from this novel method of generating super-Schur functions, one of the main motivation of this work is to make a connection with the super Kadomtsev-Petviashvili (KP) hierarchy.  
In order to substantiate this connection, we will sketch in this section the general idea.  The precise details will be postponed to a future work.

\medskip

   The KP flow \cite{Miwa00}  is generated by an Hamiltonian $H(t)=\sum_{n>0}t_n H_n$ made up of the positive modes of the current $H_n$ (the free boson modes) and of time variables  $t_n$.    
The tau-function $\tau(t;g)$ is viewed as the vacuum expectation value of the flow of the $\mathrm{GL}(\infty)$ orbit of the vacuum vector in the spin module (associated to the Clifford algebra): 
\[
\tau(t;g) = \langle 0 | \mathrm e^{H(t)} g |0\rangle, \qquad g\in \mathrm{GL}(\infty)
\]
where $|0\rangle$ denotes the vacuum.  
It is well-known that the KP hierarchy  can be formulated as a bilinear identity which is satisfied by the tau-function $\tau(t;g)$. 
In the boson-fermion correspondence, and in the language of Bernstein operators, one first defines the bilinear operator \cite{Jarvis, Goulden}
\[
S = \oint \frac{\dd z}{2\pi \ii}B(z) \otimes B^\perp(z^{-1}),
\]
and the KP equation reads
\[
S(\tau\otimes \tau)=0.
\]
This last relation is known as the Hirota bilinear identity.  In that setting, the function $\tau(t;g)$ is a symmetric function that can be expanded in the basis of Schur functions (with the time variable $t_i$ corresponding to  the power-sum $x_i$).  Knowing the tau-function from the Hirota equation, each coefficient in the expansion of $\tau(t;g)$ in the Schur basis can be obtained using  
\[
s_\la( \partial_{t_1}, \tfrac12 \partial_{t_2}, \ldots) \cdot \tau(t;g) \bigl\vert_{t=0}
\]
for any partition $\la$.

\medskip
When considering the generalization to superspace,  it is natural to consider the super-KP flow generated by a super Hamiltonian $\mathcal H(t)$ which contains the positive modes of the even current $H_n, n=1,2,\ldots$ (as above) and the positive modes of the odd current $H_n, n=\frac12,\frac32,\ldots$ (the free fermion modes), with an infinite set of even and odd time variables $t^{(0)}=(t_1, t_2,\ldots)$, $t^{(1)}=( t_{1/2}, t_{3/2}, \ldots)$.  Set $\mathcal H(t) = \sum_{n\in \frac12 \mathbb Z_{>0}} t_n H_n$.  In view of the Cauchy formula \eqref{superpartZ}, we have a natural expansion of the flow over the ring of symmetric functions in superspace in the basis of super-Schurs:
\[
\mathrm e^{\mathcal H(t)} = \sum_\La s_\La^{\phantom*}(t^{(0)};t^{(1)}) s_\La^*(H_1, \tfrac12 H_2, \ldots; H_{1/2}, H_{3/2}, \ldots).
\]
From the work of Kac \& Van De Leur \cite{KvdL,KvdL_tB}, where the authors proposed to define a super-KP hierarchy by extending the algebra to the complex free superfermion  algebra, the elements of the spin module are written as
\[
|\mathsf Y \rangle =  \psi_{r_1}^{\phantom*}\cdots \psi_{r_n}^{\phantom*} \psi_{s_1}^* \cdots \psi_{s_n}^* \, |0\rangle,
\]
where $\psi_i^{\phantom*}, \psi_j^*$ with $i,j\in \frac12\mathbb Z$  form a Clifford superalgebra, and where $r_1>\ldots> r_n>0 > s_n>\ldots>s_1$.  Our claim is that, under a superpolynomial realization, these superfermion modes are mapped to
\[
\psi_r\, \mapsto \begin{cases} B_r^{(0)} \qquad r\in \mathbb Z  \\ B_{r-\frac12}^{(1)} \qquad r\in \mathbb Z+\tfrac12 \end{cases}
\qquad \text{and} \qquad
\psi_s^*\, \mapsto \begin{cases} K_s^{\perp  (0)} \qquad s\in \mathbb Z  \\ K_{s+\frac12}^{\perp (1)} \qquad s\in \mathbb Z+\tfrac12 \end{cases}
\]
so that the state $|\mathsf Y\rangle$ is then represented by a super-Schur function
\[
\langle 0 | \mathrm e^{\mathcal H(t)} \vert \mathsf Y \rangle = s_\La(t^{(0)};t^{(1)})
\]
where the superpartition $\La$ corresponds to the superdiagram $\mathsf Y$ generated by the super Bernstein operators. By analogy, we thus consider the tau-function in superspace to be naturally expanded as
\[
\tau(t^{(0)};t^{(1)};g) = \langle 0 | \mathrm e^{\mathcal H(t)} g \vert 0 \rangle= 
\sum_\La c_\La(g)  s_\La(t^{(0)};t^{(1)}) \quad \Leftrightarrow \quad
 c_\La(g) = s_\La^*(\partial_{t_1},  \tfrac12 \partial_{t_2}, \ldots; \partial_{t_{1/2}}, \partial_{t_{3/2}}, \ldots  )\cdot \tau \bigl\vert_{t=0}
\]  
(note that the the dual family $s_\La^*$ appears in the value of the coefficients), where the $t^{(0)}$ (resp.~$t^{(1)}$) variables correspond to the super power-sum $x$ (resp.~$\ta$)  variables.

As for the bilinear identity satisfied by the super tau-function, we have that both the super Bernstein vertex operators $B(z;\eta)$ and $C^\perp(z;\eta)$ define an action on the basis of super-Schur functions of type I. The natural candidate in superspace generalizing the above operator $S$ is thus
\[
\Xi(\tau\otimes \tau)=0,
\]
where
\[
\Xi= \oint \frac{\dd z \dd \eta}{2\pi \ii} B(z;\eta)\otimes C^\perp(z^{-1};-\eta).
\]

\medskip 

Future work will address the precise connection between the super-KP hierarchy and superpolynomials.    
In particular, the following points will be investigated:
\begin{enumerate}
\item the underlying algebra satisfied by the super Bernstein operators;
\item the symmetry algebra/group to which the element $g$ in the  super tau-function belongs;
\item the super Hirota identities from the operator $\Xi$.  
\end{enumerate}

\begin{appendix}

\section{Manipulating bases of symmetric superfunctions}

\label{ExtraRels}

In this (short) appendix, we present how to obtain relations when manipulating different bases in the superpolynomials ring, in particular when partial derivatives are involved, such as 
$
(\partial_{x_m} \tilde{e}_n), \, (\partial_{\tilde{e}_m} \ta_{n}), (\tilde{e}_m^\perp \, h_n^{\phantom \perp}),$ etc.
The trick is to use  the generating series \eqref{GFhe} and \eqref{GFhee} which relate $h_n,\tilde h_n$ or $e_n,\tilde e_n$ to $x_n$ and $\ta_n$.  For instance, using \eqref{GFhee} 
\[
E(z;\eta)=   \exp \l(- \sum_{m>0} (-z)^m \l(x_m + \frac{\eta}{z} \ta_m\r)\r),
\]
 we obtain the following two equations
\be
\partial_{x_m} E(z;\eta) = (-1)^{m-1} z^m E(z;\eta), \qquad
\partial_{\ta_m}E(z;\eta)= (-1)^m z^{m-1} \eta E(z;\eta).
\ee
Substituting in these expressions the expansion into (super) elementary symmetric functions, i.e.~
\[
E(z;\eta) = \sum_{m\geq 0} z^m ( e_m + \eta \tilde{e}_m) ,
\]
and equating both sides, we then have
\be\label{dex}
\partial_{x_m}e_n = (-1)^{m-1}e_{n-m}, \qquad \partial_{x_m}\tilde{e}_n = (-1)^{m-1}\tilde{e}_{n-m},
\ee
and
\be
\partial_{\ta_m}e_n =0, \qquad \partial_{\ta_m}\tilde{e}_n = (-1)^{m-1} e_{n-m+1}.
\ee
Applying the duality $\omega$ on the previous relations, we obtain immediately
\be
\partial_{x_m}h_n = h_{n-m}, \qquad \partial_{x_m}\tilde{h}_n = \tilde{h}_{n-m}, \qquad
\partial_{\ta_m}h_n =0, \qquad \partial_{\ta_m}\tilde{h}_n =  h_{n-m+1}.
\ee
Then, consider the following generating series for the basis $\check{h}_\La$ (see eq.~\eqref{labasedeshcheck})
\[
\check{H}(z;\eta) = \exp \bigl( \sum_{m>0} z^m ( x_m -\frac{\eta}{z} (-1)^m \tilde{e}_{m-1})\bigr)
\]
from which we get the relation
\be
\partial_{\tilde{e}_m} \check{H}(z;\eta) = - \eta (-z)^m \check{H}(z;\eta).
\ee
Substituting
\[
\check{H}(z;\eta) =\sum_{m\geq 0}z^m (h_m + \eta \ta_{m+1}) 
\]
in previous equation and equating both sides, we obtain
\be\label{del_tilde_em_onta}
\partial_{\tilde{e}_m}  x_n=  \partial_{\tilde{e}_m}  h_n  =0\qquad \text{and}\qquad \partial_{\tilde{e}_m}\ta_{n+1} = (-1)^m h_{n-m}.
\ee

From the above results, it is then possible to express derivatives in  a different basis, for instance from \eqref{del_tilde_em_onta}, we have
\be\label{rel_app_del_et12}
\partial_{\tilde{e}_m} = \sum_{r>0} (\partial_{\tilde{e}_m} x_{r} )\, \partial_{x_{r}}+ \sum_{r\geq 0}(\partial_{\tilde{e}_m} \ta_{r+1} )\, \partial_{\ta_{r+1}} = (-1)^m \sum_{r\geq 0}h_r \partial_{\ta_{m+r+1}}.
\ee
Or, from \eqref{dex},
\be
\partial_{x_m} = (-1)^{m-1} \sum_{r\geq 0} \bigl( e_{r} \partial_{e_{m+r}} + \tilde{e}_r \partial_{\tilde{e}_{m+r}} \bigr).
\ee

\medskip

We end this section with the presentation of exchange relations between elementary and homogeneous superfunctions,  when one element is the adjoint map.  This may be computed directly from the generating series, using the relation: $\exp(A) \exp(B) = \exp([A,B]) \exp(B) \exp(A)$ whenever $[A,B]$ is central.   For example, we have
\be
E^\perp(z_1;\eta_1) \circ H(z_2;\eta_2) = (1+z_1z_2+\eta_1\eta_2) \, H(z_2;\eta_2) \circ E^\perp(z_1;\eta_1)
\ee
which translates into the set of identities
\be \label{rels_exch_eh}
\begin{split}
e_m^\perp\circ h_n^{\phantom \perp} & = h_n^{\phantom \perp} \circ e_m^\perp + h_{n-1}^{\phantom \perp} \circ e^\perp_{m-1} \\
\tilde{e}_m^\perp\circ h_n^{\phantom \perp} & = h_n^{\phantom \perp} \circ \tilde{e}_m^\perp + h_{n-1}^{\phantom \perp} \circ \tilde{e}^\perp_{m-1} \\
e_m^\perp\circ \tilde{h}_n^{\phantom \perp} & = \tilde{h}_n^{\phantom \perp} \circ e_m^\perp + \tilde{h}_{n-1}^{\phantom \perp} \circ e^\perp_{m-1} \\
\tilde{e}_m^\perp\circ \tilde{h}_n^{\phantom \perp} & = -\tilde{h}_n^{\phantom \perp} \circ \tilde{e}_m^\perp - \tilde{h}_{n-1}^{\phantom \perp} \circ \tilde{e}^\perp_{m-1} + h_n^{\phantom \perp} \circ e_m^\perp.
\end{split}
\ee
Likewise, we have that
\be
H^\perp(z_1;\eta_1) \circ H(z_2;\eta_2) = ({1-z_1z_2-\eta_1\eta_2})^{-1}\, H(z_2;\eta_2) \circ H^\perp(z_1;\eta_1)
\ee
implies
\be\label{hhperp}\begin{split}
h_m^\perp \circ h_n^{\phantom \perp} &=  \sum_{r\geq 0} h_{n-r}^{\phantom \perp} \circ h_{m-r}^\perp  \\
h_m^\perp \circ \tilde{h}_n^{\phantom \perp} &=  \sum_{r\geq 0} \tilde{h}_{n-r}^{\phantom \perp} \circ h_{m-r}^\perp \\
\tilde{h}_m^\perp \circ h_n^{\phantom \perp} &=  \sum_{r\geq 0} h_{n-r}^{\phantom \perp} \circ \tilde{h}_{m-r}^\perp \\
\tilde{h}_m^\perp \circ \tilde{h}_n^{\phantom \perp} &=   - \sum_{r\geq 0} \bigl( \tilde{h}_{n-r}^{\phantom \perp} \circ \tilde{h}_{m-r}^\perp - r\, h_{n-r+1}^{\phantom \perp} \circ h_{m-r+1}^\perp \bigr).
\end{split}\ee
Note that applying the duality $\omega$ on both sides of relations \eqref{rels_exch_eh} and \eqref{hhperp}, we obtain the corresponding relations of the form $h_m^\perp \circ e_n^{\phantom \perp}$ and $e_m^\perp\circ e_n^{\phantom \perp}$.

\section{The Pieri rules}
\label{PieriApp}
We present in this section the Pieri rules that are necessary
for computing the Bernstein operators $B^{(1)}_n$ and $C_n^{(0)}$ (recall that
 $B^{(0)}_n$ and $C_n^{(1)}$ can then be obtained by acting with $\partial_{\tilde e_0}$ and $\partial^{\perp}_{\tilde e_0}$ respectively).  
Explicitly, the different required Pieri rules are
\be \begin{split}
 &B^{(1)}_n: \;e_r s_\La^*\; \text{and}\; \ta_r s_\La ;
 \\
 & C^{(0)}_n: \;e_r s_\La\; \text{and}\; h_r s_\La^*.
 \end{split}
\ee
Note that the {inverse} Pieri rules, i.e. $e^\perp_r s_\La$ and $e^\perp_r s^*_\La$ that
 are needed respectively for computing $B_n^{(1)}$ and $C_n^{(0)}$, are replaced by their dual using the relations  
\be
\langle e_r^\perp s_\La^{\phantom *},s^*_\Omega\rangle=
\langle s_\La^{\phantom *}, e_r s^*_\Omega\rangle 
\qquad\text{and}\qquad 
\langle e_r s_\La^{\phantom *},s^*_\Omega\rangle=
\langle s_\La^{\phantom *}, e_r^\perp s^*_\Omega\rangle.
\ee  
We now describe the different Pieri rules.\footnote{The first three rules are illustrated (and proved) in \cite{JL17}. For further examples of the first two, see \cite{BM1}.}

\subsection{The Pieri rule $e_r s_\La$}
We have
\be
e_r s_{\La} = \sum_\Omega  s_{\Omega},
\ee
where the sum is over all superpartitions $\Omega$ of fermionic degree $\mathrm{m}(\La)$ such that:
\begin{itemize} 
\item $\Om^*/\La^*$ is a vertical $r$-strip.

\item The 
circles of $\La$ can be moved subject to the following restrictions:
\begin{enumerate}
\item[(i)] a circle in the first column can be moved vertically without 
restrictions;
\item[(ii)] a circle not in the first column can be moved vertically in the same column as 
  long as there is a square in the column
  immediately to its left it in the original 
diagram $\La$;
\item[(iii)] a circle can be moved horizontally in the same row by at most one column. 
\end{enumerate}

\end{itemize}

\subsection{The Pieri rule $\ta_r s_\La$}
We have
\be\label{Pierithetar}
\ta_r s_{\La} = \sum_\Omega (-1)^{\#\ell(\circledast)} s_{\Omega},
\ee
where $\#\ell(\circledast)$ denotes the number of circles above the one added in the diagram
and
where  the sum is over all superpartitions $\Omega$ of fermionic degree $\mathrm{m}(\La)+1$ such that:
\begin{itemize}
\item $\Om^*/\La^*$ is a horizontal $(r-1)$-strip and when $\Om/\La$ is straightened horizontally the circle is in the rightmost position.

\item The 
circles of $\La$ can be moved subject to the following restrictions:
\begin{enumerate}
\item[(i)] a circle in the first row can be moved horizontally without 
restrictions;
\item[(ii)] a circle not in the first row can be moved horizontally along the same row  as 
long as there is a square in the row just above it in the original 
diagram $\La$;
\item[(iii)] a circle can be moved vertically in the same column by at most one row.
\end{enumerate}

\end{itemize}

\subsection{The Pieri rule $h_rs_\La^*$} We have
\be\label{Pierihset123}
 h_r s^*_{\La} = \sum_\Omega  s^*_{\Omega},
\ee
where 
 the sum is over all superpartitions $\Omega$ of fermionic degree $\mathrm{m}(\La)$ such that:
\begin{itemize}
 
\item $\Om^*/\La^*$ is a horizontal $r$-strip.

\item
The $i$-th circle, starting from below, of $\Om$ is either in the same row  as the
$i$-th circle of $\La$ if $\Om^*/\La^*$ does not contain a box in that row or one row below that  of the $i$-th circle of $\La$ if $\Om^*/\La^*$  contains a box in the row  of the $i$-th circle of $\La$.

\end{itemize}

\subsection{The Pieri rule $e_rs_\La^*$}
The following Pieri rule, which has been overlooked in previous works, is presented with its proof: 
\be\label{PieriRRtIII1}e_r s^*_\La=\sum_{\Omega}
 {s}_\Omega^*
\ee
 where the sum is over all superpartitions $\Omega$ of fermionic degree $\mathrm{m}(\La)$ such that:
\begin{itemize}
\item $\Om^*/\La^*$ is a vertical $r$-strip 

\item The 
circles of $\La$ can be moved subject to the following restrictions:
\begin{enumerate}

\item[(i)] a circle cannot overpass another one;

\item[(ii)]  a circle cannot be moved in a row which has an added box;

\item[(iii)] {the addition of  a bosonic box to a fermionic row bumps the circle to the end of the subsequent row (if it is bosonic)  and this bumping can be done repeatedly.}

\end{enumerate}
\end{itemize}

\medskip
\noindent
For example, marking the added boxes with a $\times$, we have
\[
e_{2} \cdot
\scalebox{1.1}{$\superY{
\,&\,& \yF{}\\
\, 
}$}
 =
 \scalebox{1.1}{$\superY{
\,&\,& \yF{}\\
\, \\
\times\\
\times
}$}
+
\scalebox{1.1}{$\superY{
\,&\,& \yF{}\\
\, & \times \\
\times
}$}
+
\scalebox{1.1}{$\superY{
\,&\,& \times\\
\, & \yF{}\\
\times
}$}
+
\scalebox{1.1}{$\superY{
\,&\,& \times\\
\, & \times\\
\yF{}
}$}.
\]

\subsection{Proof of the Pieri rule $e_rs_\La^*$} \label{B5}
\begin{proof}  We proceed by induction.   The case $e_1 s_\La^*$  holds since $e_1=h_1$ and the Pieri rule for $h_1s_\La^*$ stated in
    \eqref{Pierihset123}   is compatible with \eqref{PieriRRtIII1} in the $e_1$ case.  We now assume that the Pieri rule is valid for $e_r, e_{r-1}, \ldots, e_1$ and use the relation \cite{MacSym95}
      \be
e_{r+1} +   \sum_{s=1}^{r+1} (-1)^{s} h_{s} \,  e_{r+1-s}=0   
  \ee
  to compute the $e_{r+1}$ case.  Since the previous relation uniquely determines $e_{r+1}$, it suffices to show 
    that our Pieri rules are such that 
 \be \label{eqaprouver}
e_{r+1} s_\La^* +   \sum_{s=1}^{r+1} (-1)^{s} h_{s} \circ e_{r+1-s}   \, s_\La^*=0\, .
\ee
 In order to prove \eqref{eqaprouver}, we will first 
introduce a notation that associates a diagram to each super-Schur functions
that arise when computing the  Pieri rules.  We will then construct  
a sign-reversing involution  that cancels the diagrams two-by-two.  

\medskip

When acting with the elementary symmetric function $e_{r+1-s}$ on $s_\La^*$, for $0\leq s\leq r+1$, we mark every box added to the original superpartition $\La$ by the symbol $\times$.  Here is a simple example of such a diagram 
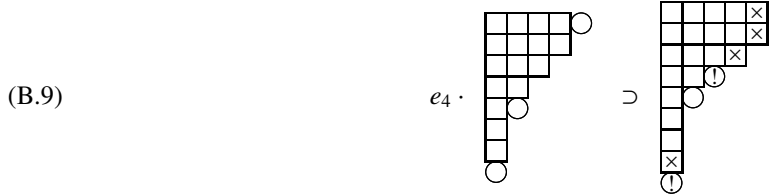
\be
e_{4} \cdot
\scalebox{1.1}{$\superY{
\,&\,&\,&\,&\yF{}\\
\,&\,&\,&\, \\
\,&\,&\, \\
\,&\, \\
\,&\yF{}\\
\,\\
\,\\
\yF{}
}$}
\supset
\scalebox{1.1}{$\superY{
\,&\,&\,&\,&\times \\
\,&\,&\,&\,&\times \\
\,&\,&\,&\times \\
\,&\,&\yF{!} \\
\,&\yF{}\\
\,\\
\,\\
\times\\
\yF{!}
}$}
\ee
(we indicate with a "!" the circles that were moved from the original diagram).  

\medskip

Then, when acting with $h_{s+1}$, we will mark the new boxes with the symbol $+$. Hence, each diagram is decorated with the symbols $\times,+$ (and ``!''), where the boxes with a $\times$ form a vertical strip while  those with a $+$ form a horizontal strip.  
For example, we have
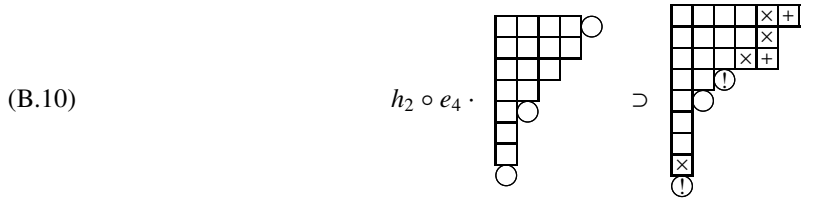
\be
h_2\circ e_{4} \cdot
\scalebox{1.1}{$\superY{
\,&\,&\,&\,&\yF{}\\
\,&\,&\,&\, \\
\,&\,&\, \\
\,&\, \\
\,&\yF{}\\
\,\\
\,\\
\yF{}
}$}
\supset
\scalebox{1.1}{$\superY{
\,&\,&\,&\,&\times & + \\
\,&\,&\,&\,&\times \\
\,&\,&\,&\times&+ \\
\,&\,&\yF{!} \\
\,&\yF{}\\
\,\\
\,\\
\times\\
\yF{!}
}$}.
\label{EX}
\ee

\medskip 

We now define an involution $\Upsilon$ acting on the diagrams.
Reading  the diagram from bottom left  to top right,  $\Upsilon$  
takes the first decorated box and, depending on its content,
changes a $\times$ into a  $+$  or a $+$ into a $\times$.
The application $\Upsilon$, which is obviously an involution,
is well defined since the $\times$'s (resp. the $+$'s) will
still form a vertical strip (resp. horizontal strip) by definition of the
first decorated box.  Moreover, the circles marked with a  ``!'' are unaffected
since both $\times$ and $+$ move a circle to the next row (although the $\times$
does it in a transitory way, the $+$ will then, if it applies, continue the movement of the circle to the next row).
The action of $\Upsilon$ results in pairing two diagrams having the same shape and only differing by the marking of one $\times/+$ box. As such, they will 
appear with   
opposite sign in \eqref{eqaprouver} and thus cancel each other. 
For instance, the involution transforms the diagram to the right of \eqref{EX} as follows:
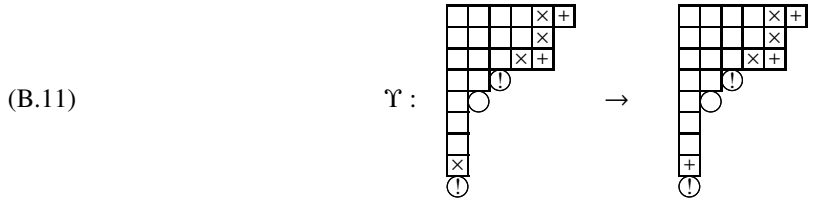
\be
\Upsilon:\;
\scalebox{1.1}{$\superY{
\,&\,&\,&\,&\times & + \\
\,&\,&\,&\,&\times \\
\,&\,&\,&\times&+ \\
\,&\,&\yF{!} \\
\,&\yF{}\\
\,\\
\,\\
\times\\
\yF{!}
}$}
\rightarrow
\quad
\scalebox{1.1}{$\superY{
\,&\,&\,&\,&\times & + \\
\,&\,&\,&\,&\times \\
\,&\,&\,&\times&+ \\
\,&\,&\yF{!} \\
\,&\yF{}\\
\,\\
\,\\
+\\
\yF{!}
}$}.
\ee
   The resulting diagram is among those obtained from the action of $h_3\circ e_3$ on the original diagram.    If these two diagrams appeared in the sum, they would differ by a relative sign and cancel each other.
\end{proof}

We end this section with a simple example illustrating the involution $\Upsilon$.  Consider the Pieri rule $e_3s_{(2;1)}^*$.  We have to show that

\be \label{aprouver2}
\bigl(e_3- h_1 \circ e_2 + h_2 \circ e_1 - h_3 \bigr) 
  \scalebox{1.1}{$\superY{
	\, &\,& \yF{}\\
	\,  
	}$}=0
  \ee
For the first term, we have
\be\label{line0}
\begin{split}
& 
\scalebox{1.1}{$\superY{
	\, &\,& \yF{}\\
	\,  \\
	\times \\
	\times \\
	\times
	}$}
+  
\scalebox{1.1}{$\superY{
	\, &\,& \yF{}\\
	\, & \times \\
	\times \\
	\times
	}$}
+  
\scalebox{1.1}{$\superY{
	\, &\,& \times\\
	\, &  \yF{} \\
	\times \\
	\times
	}$}	
+  
\scalebox{1.1}{$\superY{
	\, &\,& \times\\
	\, &  \times \\
	\times \\
	\yF{}
	}$}	
	.
\end{split}
\ee
The second term reads
\be\label{line1}
\begin{split}
& 
-\scalebox{1.1}{$\superY{
	\, &\,& \yF{}\\
	\,  \\
	\times \\
	\times \\
	+
	}$}
-
\scalebox{1.1}{$\superY{
	\, &\,& \yF{}\\
	\, & \times \\
	\times \\
	+
	}$}
-  
\scalebox{1.1}{$\superY{
	\, &\,& \times\\
	\, &  \yF{} \\
	\times \\
	+
	}$}	
-  
\scalebox{1.1}{$\superY{
	\, &\,& \times\\
	\, &  \times \\
	+ \\
	\yF{}
	}$}	
\\
&
-	 
\scalebox{1.1}{$\superY{
	\, &\,& +\\
	\, &  \yF{} \\
	\times \\
	\times
	}$}	
-
\scalebox{1.1}{$\superY{
	\, &\,& \yF{}\\
	\, & + \\
	\times \\
	\times
	}$}
-
\scalebox{1.1}{$\superY{
	\, &\,& + \\
	\, & \times & \yF{} \\
	\times
	}$}	
-
\scalebox{1.1}{$\superY{
	\, &\,& \yF{}\\
	\, & \times \\
	\times &+ 
	}$}
-
\scalebox{1.1}{$\superY{
	\, &\,& \times & + \\
	\, &  \yF{} \\
	\times
	}$}	
-
\scalebox{1.1}{$\superY{
	\, &\,& \times \\
	\, &  + \\
	\times &  \yF{} 
	}$}	
\\
&
-
\scalebox{1.1}{$\superY{
	\, &\,& \times & + \\
	\, &  \times \\
	\yF{}
	}$}	
-
\scalebox{1.1}{$\superY{
	\, &\,& \times  \\
	\, &  \times & + \\
	\yF{}
	}$}	.
\end{split}
\ee
The third term reads
\be\label{line2}
\begin{split}
&	 
\scalebox{1.1}{$\superY{
	\, &\,& +\\
	\, &  \yF{} \\
	\times \\
	+
	}$}	
+
\scalebox{1.1}{$\superY{
	\, &\,& \yF{}\\
	\, & + \\
	\times \\
	+
	}$}
+
\scalebox{1.1}{$\superY{
	\, &\,& + \\
	\, & \times & \yF{} \\
	+
	}$}	
+
\scalebox{1.1}{$\superY{
	\, &\,& \yF{}\\
	\, & \times \\
	+ &+ 
	}$}
+
\scalebox{1.1}{$\superY{
	\, &\,& \times & + \\
	\, &  \yF{} \\
	+
	}$}	
+
\scalebox{1.1}{$\superY{
	\, &\,& \times \\
	\, &  + \\
	+ &  \yF{} 
	}$}	
\\
&
+
\scalebox{1.1}{$\superY{
	\, &\,& \times & + \\
	\, &  + \\
	\yF{}
	}$}	
+
\scalebox{1.1}{$\superY{
	\, &\,& \times  \\
	\, &  + & + \\
	\yF{}
	}$}	
\\
&
+
\scalebox{1.1}{$\superY{
	\, &\,& \times & + & +  \\
	\, &\yF{}   
	}$}	
+
\scalebox{1.1}{$\superY{
	\, &\, & + & +  \\
	\, & \times & \yF{}   
	}$}
+
\scalebox{1.1}{$\superY{
	\, &\, & + & +  \\
	\, & \yF{} \\
	\times   
	}$}
+
\scalebox{1.1}{$\superY{
	\, &\, & +   \\
	\, & + & \yF{} \\
	\times   
	}$}.	
\end{split}
\ee
Finally, the last term yields
\be \label{line3}
-\scalebox{1.1}{$\superY{
	\, &\,& + & + & +  \\
	\, &\yF{}   
	}$}	
-
\scalebox{1.1}{$\superY{
	\, &\, & + & +  \\
	\, & + & \yF{}   
	}$}
-
\scalebox{1.1}{$\superY{
	\, &\, & + & +  \\
	\, & \yF{} \\
	+   
	}$}
-
\scalebox{1.1}{$\superY{
	\, &\, & +   \\
	\, & + & \yF{} \\
	+   
	}$}.	
\ee
We see that the four diagrams in \eqref{line0} cancel the first four diagrams
in \eqref{line1}. The eight diagrams in the second and third lines of \eqref{line1} cancel the eight diagrams in the first two lines of \eqref{line2}.
Finally, the four diagrams in the third line of \eqref{line2} cancel the four diagrams of \eqref{line3}.  Thus, assuming that the Pieri rules for $e_1$ and $e_2$ hold, we have that the $e_3$ rule is exactly the one needed in order for \eqref{aprouver2} to be satisfied. 

\medskip

\subsection{The action $\partial_{\tilde{e}_0}$ on Schur functions}
We have seen that $B_n^{(0)}$ and $C_n^{(1)}$ are related to
 $B_n^{(1)}$ and $C_n^{(0)}$ in the following way
\be
B_n^{(0)} = \partial_{\tilde{e}_0} B_n^{(1)}, \qquad C_n^{(1)} = \partial_{\tilde{e}_0}^\perp C_n^{(0)} \, .
\ee
In the $B_n^{(0)}$  case, we thus need the rule for the action of $\partial_{\tilde{e}_0}$ on the super-Schur function $s_\La$. But by duality, 
 this is equivalent to the action of $\partial_{\tilde{e}_0}^\perp$ on $s_\La^*$ needed in the $C_n^{(1)}$  case.   
 \medskip

The operator $\partial_{\tilde{e}_0}$ turns out to have a very simple action on
the $s_\Lambda$'s. 
\begin{lemma} \label{Lemlemlemdeleo}
We have 
\be \label{lemdele0se0}
\partial_{\tilde{e}_0} s_\La = 0 {\rm ~~if~} \La_1^* = \La_1^\circledast
\qquad  {\rm while}  \qquad 
\partial_{\tilde{e}_0} s_\Lambda= s_\Omega {\rm ~~if~} \La_1^* \neq \La_1^\circledast 
\ee
where $\Omega = (\La_2^{\rm a}, \ldots; \La_1^{\rm a}, \La_1^{\rm s}, \ldots)$, that is, where $\Omega$ is the superpartition obtained by removing the circle in the first row of the diagram of $\Lambda$.  
\end{lemma}
\begin{proof}
  Define the lexicographic order on superpartitions of the same fermionic and total degrees 
  to be such that
  $\Lambda >_l \Omega$ if 
  the smallest $i$ in which row $i$ of the diagrams
  of $\Lambda$ and $\Omega$ differ is such that
  $\Lambda_i^* > \Om_i^* $ or $\La_i^\circledast  > \Om_i^\circledast$.
  The proof proceeds by induction on the fermionic degree, the total degree and the reverse lexicographic order.
 The proof will make use of 
\eqref{rel_app_del_et12} in the case $m=0$:
  \be \label{eqe0}
\partial_{\tilde e_0} = \sum_{r \geq 0} h_r \partial_{\theta_{r+1}}. 
\ee

  The result obviously holds for superpartitions $\Lambda$ of fermionic degree 0.  Now, fix a given fermionic degree ${\rm m}>0$ and a given total degree $d$. 
  Suppose that $\Lambda$ is the largest superpartition (in lexicographic order) of fermionic degree ${\rm m}$ and total degree $d$, that is, $\Lambda$ is of the form $\La=(\ell,\m-2,\m-3,\dots,1,0;\,)$,
  where $\ell=d-({\rm m}-1)({\rm m}-2)/2$.   We will show that the result holds in this case, that is, that $\partial_{\tilde{e}_0} s_\La= s_\Omega$ where
  $\Omega=(\m-2,\m-3,\dots,1,0;\ell)$.  
  There are two cases to consider: $\m=1$ and $\m >1$.  If $\m=1$, we have that $s_\Lambda=\theta_{d+1}$ (see Table \ref{tab1} for instance) which, by \eqref{eqe0}, leads to $\partial_{\tilde{e}_0} s_\La= h_d=s_\Omega$. 
On the other hand, if $\m>1$ then
$\tilde e_0 s_\Omega=s_\Lambda$ by the Pieri rule \eqref{Pierithetar} (recall that $\tilde{e}_0= \ta_1$) since the new circle can only be put in the first row.  This implies that
  \be
  \partial_{\tilde{e}_0} s_\La= \partial_{\tilde{e}_0} \tilde e_0 s_\Omega =
  s_\Omega - \tilde e_0 \partial_{\tilde{e}_0}  s_\Omega = s_\Omega
  \ee
since $\partial_{\tilde{e}_0}  s_\Omega=0$ by induction on the fermionic degree.  Thus $\partial_{\tilde{e}_0} s_\La= s_\Omega$ when $\Lambda$ is the largest superpartition in lexicographic order for fixed fermionic and total degrees.

 For a given fermionic degree ${\rm m}$,  the smallest possible total degree
  is $d=\m(\m-1)/2$ with  $\Lambda=(\m-1,\m-2\dots,1,0;\,)$ being the only superpartition of that degree.  Since this superpartition is also the largest one of its degree in lexicographic order, the result holds for the lowest total degree since we have already proven that it holds for the largest superpartition in lexicographic order for fixed fermionic and total degrees.

  We can now proceed to the final (but by far most technical) step of the proof.  We will show that the result holds for an arbitrary superpartition $\Lambda$ of fermionic degree ${\mathrm m}$
  and total degree $d$.  
  By induction, we can suppose that the result holds for every superpartition 
  \begin{enumerate}
  \item[$\bullet$]  of fermionic degree smaller than ${\rm m}$;
  \item[$\bullet$]  of total degree smaller than $d$;
  \item[$\bullet$] or larger than $\Lambda$ with respect to the lexicographic order.
  \end{enumerate}
  We will treat separately the two possible cases, i.e.~$\Lambda_1^* \neq \Lambda_1^\circledast$,  and $\Lambda_1^* = \Lambda_1^\circledast$.  
  First consider the case where $\Lambda$ is such that $\Lambda_1^* \neq \Lambda_1^\circledast$.
 Let $\tilde{\Lambda}$ be the superpartition whose diagram is that of $\Lambda$ without its first part and set $\ell=\Lambda_1^\circledast$.   Using the Pieri rule, we have
  \be \label{pieri1}
  \theta_\ell s_{\tilde \Lambda} = s_{\Lambda} + \sum_{\Gamma \neq \Lambda} s_{\Gamma}
  \ee
  where the superpartitions $\Gamma$ appearing in the sum are given by the  Pieri rule \eqref{Pierithetar}, and are such that $\Gamma_1^\circledast > \Lambda_1^\circledast =\ell$ since $\Gamma^\circledast/\tilde \Lambda^\circledast$ is a vertical $\ell$-strip and $\Gamma \neq \Lambda$ (given that $\tilde \Lambda_1^\circledast < \ell$, there is a unique way to add $\ell$ boxes in the first $\ell$ columns and it gives rise to $\Lambda$). This is illustrated in the following example where the diagrams appearing in the expansion of  $\theta_4$ acting on  $s_{(1;3)}$ 
  \be
\scalebox{1.1}{$\superY{
\, & \, &\, & \times & \times &  \times & \yF{\times}\\
\,& \yF{}
}$}\quad \scalebox{1.1}{$\superY{
\, & \, &\, & \times &  \times & \yF{\times}\\
\,& \times &\yF{}
  }$}
\quad \scalebox{1.1}{$\superY{
\, & \, &\, & \times &  \times & \yF{\times}\\
\, &\yF{} \\
\times
  }$}
\quad \scalebox{1.1}{$\superY{
\, & \, &\,   &  \times & \yF{\times}\\
\, & \times &\yF{} \\
\times
  }$}
\quad \scalebox{1.1}{$\superY{
\, & \, &\,   &  \times & \yF{\times}\\
\, & \times  \\
\times & \yF{}
  }$}
\quad \scalebox{1.1}{$\superY{
\, & \, &\,    & \yF{\times}\\
\, & \times & \times  \\
\times & \yF{}
  }$}
\ee
all have a new circle at the end of the first row and the only one whose first row is of length 4 (including the circle) is the last diagram, $\Lambda=(3,1;3)$.
Hence all $\Gamma$'s that appear in the sum in \eqref{pieri1} are
larger than $\Lambda$ in lexicographic order.  Note that the sign is always positive because the new circle needs to be at the end of the first row of $\Gamma$
(the new circle is always the rightmost among the cells of $\Gamma^\circledast/\tilde \Lambda^\circledast$ as can be seen in the previous example).
Therefore, we write
 \be
 s_{\Lambda} =  \theta_\ell s_{\tilde \Lambda} - \sum_{\Gamma >_l \Lambda}  s_{\Gamma}
  \ee
 where (recall that) the sum is over superpartitions from the Pieri rule \eqref{Pierithetar} and such that $\Gamma >_l \Lambda$.  
  This gives, from \eqref{eqe0}, that
   \be \label{Beq1}
\partial_{\tilde e_0}  s_{\Lambda} =  h_{\ell-1} s_{\tilde \Lambda}-\theta_\ell \partial_{\tilde e_0}  s_{\tilde \Lambda}  - \sum_{\Gamma >_l \Lambda }  \partial_{\tilde e_0} s_{\Gamma} =    h_{\ell-1} s_{\tilde \Lambda}-\theta_\ell \partial_{\tilde e_0}  s_{\tilde \Lambda}  - \sum_{\Gamma >_l \Lambda }  s_{\hat \Gamma}
\ee
where $\hat \Gamma$ is $\Gamma$ without the circle in its first row (we used the fact that by induction on reverse lexicographic order
$\partial_{\tilde e_0}  s_{\Gamma}=s_{\hat \Gamma}$).  
Because of the term $\partial_{\tilde e_0}  s_{\tilde \Lambda}$ in the previous equation, 
we need at this point to consider the two cases   $\tilde \Lambda_1^*\neq \tilde \Lambda_1^\circledast$ and $\tilde \Lambda_1^* = \tilde \Lambda_1^\circledast$ separately.  Suppose that  $\tilde \Lambda_1^*\neq \tilde \Lambda_1^\circledast$.  By induction on the fermionic degree, we have  $\partial_{\tilde e_0}  s_{\tilde \Lambda}= s_{\tilde \Omega}$ where
${\tilde \Omega}$ is $\tilde \Lambda$ without its circle in the first row,  and equation \eqref{Beq1} becomes 
  \be \label{b24}
  \partial_{\tilde e_0}  s_{\Lambda}
  = h_{\ell-1} s_{\tilde \Lambda}-\theta_\ell  s_{\tilde \Omega}  - \sum_{\Gamma >_l \Lambda }  s_{\hat \Gamma}
  \ee
  Now, $h_{\ell-1} s_{\tilde \Lambda}$ can give rise to $s_\Delta$ where $\Delta$ has
  a circle in the first row or not.  We will show that each $\Delta$ that has a circle in the first row will be canceled out by an equal term coming from $\ta_\ell s_{\tilde \Omega}$; each  $\Delta$ that has no circle in the first row will be canceled  out by an equal term coming from $ \sum_{\Gamma >_l \Lambda }  s_{\hat \Gamma}$, except the desired case $\Delta=\Omega$.  
  (Recall that there is no superpartition $\Gamma$ in the summation such that $\hat \Gamma=\Omega$ since this would imply that  $\Gamma$ must be $\La$ itself, hence a contradiction.)  
  Consider for instance the case where $h_3$ acts on $s_{(2;1)}$ which gives rise to the following diagrams:
  \be \label{example2}
\scalebox{1.1}{$\superY{
\, & \,  & \times & \times &  \times & \yF{}\\
\,
}$}\quad \scalebox{1.1}{$\superY{
\, & \, & \times &  \times & \yF{}\\
\,& \times 
  }$}
\quad \scalebox{1.1}{$\superY{
\,  &\, & \times &  \times & \yF{}\\
\,  \\
\times
  }$}
\quad \scalebox{1.1}{$\superY{
\,  &\,   &  \times & \yF{}\\
\, & \times  \\
\times
  }$}
\quad \scalebox{1.1}{$\superY{
\, & \, &\times   &  \times \\
\, & \times & \yF{}
  }$}
\quad \scalebox{1.1}{$\superY{
\, & \, &    \times\\
\, & \times & \yF{}  \\
\times 
  }$}
\ee
  If $\Delta$ has a circle in the first row, then
  it corresponds to one of the terms appearing in $\theta_\ell  s_{\tilde \Omega}$,
  which is the case in our example where the first 4 terms correspond to
  $\theta_4 \,  s_{(\,;2,1)}= s_{(5;1)}+s_{(4;2)}+s_{(4;1,1)}+s_{(3:2,1)}$.
This is seen in the following way:
since the new circle in  $\theta_\ell  s_{\tilde \Omega}$
needs to appear at the end of the first row by the argument given previously,  acting with $\theta_\ell$  amounts to acting with $h_{\ell-1}$  and then adding a circle at the end of the first row, which is exactly the same as acting with $h_{\ell-1}$ on  $s_{\tilde \Lambda}$ in the case where the circle in the first row is pushed along its row.  
If $\Delta$ has no circle in the first row then the circle in the first row has been pushed to the second row (see the last two diagrams in the example).  But this amounts to acting with $\theta_\ell$ on $s_{\tilde \Lambda}$
and then removing the circle in the first row which, except for $\Delta =\Omega$, are exactly the terms appearing in the sum in \eqref{b24} (compare the last two diagrams in the previous example with the superpartitions obtained after removing the circle in the first row of the superpartitions appearing in the expansion  
$\theta_4 \,  s_{(2;1)}= s_{(4,2;\,)}+s_{(3,2;1)}$).
The only term remaining on the right-hand side of \eqref{b24} is thus $s_\Omega$ which means that $ \partial_{\tilde e_0}  s_{\Lambda} = s_\Omega$ holds in that case.  

 Going back to the equation \eqref{Beq1}, we now suppose the case $\tilde \Lambda_1^*= \tilde \Lambda_1^\circledast$. 
  By induction on the fermionic degree, we then have that 
$\partial_{\tilde e_0}  s_{\tilde \Lambda}= 0$,  and thus
 \be 
  \partial_{\tilde e_0}  s_{\Lambda}
  = h_{\ell-1} s_{\tilde \Lambda} - \sum_{\Gamma >_l \Lambda }  s_{\hat \Gamma}.
  \ee
  But again  acting with $h_{\ell-1}$ on  $s_{\tilde \Lambda}$
  amounts to acting with $\theta_\ell$  on  $s_{\tilde \Lambda}$
  and then removing the circle in the first row, which are exactly the terms appearing in the sum on the right-hand side,  except for $\hat \Gamma =\Omega$ (as explained above).   This proves again that
$ \partial_{\tilde e_0}  s_{\Lambda} = s_\Omega$  in that case.

\medskip 
We now turn to the other possible case, namely when $\Lambda$ is such that $\Lambda_1^* = \Lambda_1^\circledast$.   This case is similar to the previous one but somewhat easier.  We let  once more $\tilde \Lambda$ be the superpartition whose diagram is that of $\Lambda$ without its first part. Setting $\Lambda_1^\circledast= \ell$, we observe from the Pieri rule that
  \be \label{pieri2} 
  h_\ell s_{\tilde \Lambda} = s_{\Lambda} + \sum_{\Gamma \neq \Lambda} s_{\Gamma}
  \ee
  where the superpartitions $\Gamma$ appearing in the sum are all such that $\Gamma_1^\circledast > \Lambda_1^\circledast =\ell$ since $\Gamma^\circledast/\tilde \Lambda^\circledast$ is a vertical $\ell$-strip and $\Gamma \neq \Lambda$ (given that $\tilde \Lambda_1^\circledast \leq \ell$, there is again a unique way to add $\ell$ boxes in the first $\ell$ columns and it also gives rise to $\Lambda$.
  This is illustrated with $h_3$ acting on $s_{(2;1)}$ in \eqref{example2} where
all the diagrams have a first row of length larger than 4 except 
$\Lambda=(2;3,1)$).
Therefore
 \be \label{pieri3}
 \partial_{\tilde e_0} s_{\Lambda}=   h_\ell  \partial_{\tilde e_0} s_{\tilde \Lambda} - \sum_{\Gamma >_l \Lambda}  \partial_{\tilde e_0} s_{\Gamma}
 \ee
since from our previous observation all $\Gamma$'s that appear in the sum are
larger than $\Lambda$ in lexicographic order.
We  need as before to consider the two cases   $\tilde \Lambda_1^*\neq \tilde \Lambda_1^\circledast$ and $\tilde \Lambda_1^* = \tilde \Lambda_1^\circledast$ separately.  
First 
suppose that  $\tilde \Lambda_1^*\neq \tilde \Lambda_1^\circledast$.  By induction on the total degree, we have  $\partial_{\tilde e_0}  s_{\tilde \Lambda}= s_{\tilde \Omega}$, where
${\tilde \Omega}$ is $\tilde \Lambda$ without its circle in the first row.
By induction again, but this time on the reverse lexicographic order, we obtain $\sum_{\Gamma >_l \Lambda}  \partial_{\tilde e_0} s_{\Gamma}= \sum_{\Delta } s_{\hat \Delta}$, where the sum is over all the superpartition $\Delta$'s such that  $s_\Delta$  appears  in the product $h_\ell s_{\tilde \Lambda}$ and that $\Delta$ has a circle in the first row (otherwise it gets killed by $\partial_{\tilde e_0}$ by induction).  The superpartition $\hat \Delta$ denotes $\Delta$ without the circle in its first row.  
Using again $h_3$ acting on $s_{(2;1)}$ as an example, the $\Delta$'s are the first four diagrams in \eqref{example2} and $\hat \Delta$ are those same four diagrams but without the circle in the first row.
Hence,
 \be 
 \partial_{\tilde e_0} s_{\Lambda}=   h_\ell  s_{\tilde \Omega} - \sum_{\Delta}  s_{\hat \Delta} .
  \ee
  But, as we will see, we have precisely $h_\ell  s_{\tilde \Omega} = \sum_{\Delta}  s_{\hat \Delta}$.  
  One can compare for instance
  $h_3 s_{(\,;2,1)}=s_{(\,;5,1)} + s_{(\,;4,2)}+ s_{(\,;4,1,1)} + s_{(\,;3,2,1)}$ with the first four diagrams in \eqref{example2} but without the circle in the first row.  This is because
  acting with
  $h_\ell$ on $s_{\tilde \Omega}$ is equivalent to adding a circle in the first row of
  $\tilde \Omega$ to obtain $\tilde \Lambda$, then acting with $h_\ell$ on $s_{\tilde \Lambda}$ to get the $s_\Gamma$'s, then considering only the terms in the expansion whose circle is pushed along the first row (the $s_\Delta$'s)
  and then removing that same exact circle in the first row to get the $s_{\hat \Delta}$'s (which is the same as never having considered the circle in the
  first row of $\tilde \Lambda$).  Thus $\partial_{\tilde e_0} s_{\Lambda}=0$ holds in that case.

  Suppose finally that $\tilde \Lambda_1^*=\tilde \Lambda_1^\circledast$. In that case, the $\Gamma$'s appearing
 in the sum in \eqref{pieri2} are all such that $\Gamma_1^*=\Gamma_1^\circledast$ since $h_\ell$ does not add a new circle and
 $\tilde \Lambda$ does not have a circle in its first row that could be pushed along the first row.  Therefore, as wanted, \eqref{pieri3} leads to
\be 
 \partial_{\tilde e_0} s_{\Lambda}=   h_\ell  \partial_{\tilde e_0} s_{\tilde \Lambda} - \sum_{\Gamma >_l \Lambda}  \partial_{\tilde e_0} s_{\Gamma}=0
  \ee
  since $\partial_{\tilde e_0} s_{\tilde \Lambda}=0$ and  $\partial_{\tilde e_0} s_{\Gamma}=0$ by induction on the total degree and the reverse lexicographic order respectively.  
 This complete the proof.  
   \end{proof}

\begin{corollary}\label{corde0tiladj_1}
We have 
\be \label{lemdele0se0_adj1}
\partial_{\tilde{e}_0}^\perp s_\La^* = 0 {\rm ~~if~} \La_1^* \neq \La_1^\circledast
\qquad  {\rm while}  \qquad 
\partial_{\tilde{e}_0}^\perp s_\Lambda^*= s_\Omega^* {\rm ~~if~} \La_1^* = \La_1^\circledast 
\ee
where $\Omega = (\La_1^{\rm s}, \La_1^{\rm a}, \ldots;  \La_2^{\rm s}, \ldots)$, that is, where $\Omega$ is the superpartition obtained by adding a circle in the first row of the diagram of $\Lambda$.  
\end{corollary}
\begin{proof} This is a direct consequence of the duality between the $s_\La^{\phantom *}$ and the $s^*_\La$ bases. Suppose we want to find which coefficients $u_{\La;\Om}$ in $\partial_{\tilde{e}_0}^\perp s_\La^* = \sum_\Om u_{\La;\Om} s^*_\Om$ are non-zero.  Using the duality and Lemma \eqref{lemdele0se0}, we find
$u_{\La;\Om} = 
\delta_{\La, \hat{\Om}} 
$ whenever  $ \Om_1^* \neq \Om_1^\circledast $, and
$u_{\La;\Om} =0$  whenever $ \Om_1^* = \Om_1^\circledast $, where $\hat{\Om}$ denotes the superpartition $\Om$ without its circle in the first row.   
\end{proof}

\section{Proof of Proposition \ref{prop_typeI}}
\label{Sproof_typeI}

In this appendix, we provide the proof of Proposition \ref{prop_typeI} for the super Bernstein operators associated to super-Schur  of type I.  The proof is almost identical to the one of Section~\ref{Sproof_typeIe}, the only difference  being the type of Pieri rules involved.    

\medskip

First observe that when there are no fermionic parts, the expression
\be
B_{k_1}^{(0)} \ldots B_{k_n}^{(0)} \cdot 1
\ee
gives the usual Schur function $s_\la$ associated to the partition $\la=(k_1, \ldots, k_n)$.  Indeed, we see from the expression \eqref{Bm0} that  the second term of each  $B_k^{(0)}$ vanishes when acting on an expression which contains no odd variables.   
As a result, each of the $B_k^{(0)}$ becomes precisely a usual Bernstein operator.  On the other hand, if there is a single fermionic part, the proposition is also verified:
\be
B_k^{(1)} \cdot 1 = \sum_{r\geq 0} (-1)^r \ta_{k+r+1}^{\phantom \perp} \circ e_r^\perp \cdot 1 = \ta_{k+1}
\ee
since only the term $r=0$ contributes, and by definition $s_{(k;)}=\ta_{k+1}$.  

\medskip

Proving Proposition \ref{prop_typeI} amounts to showing that the action of a  mode $B_n^{(\epsilon)}$ on a super-Schur $s_\La$, with $n\geq \La_1^\circledast$  gives   a single super-Schur. The two cases, corresponding to the two values of $\epsilon$, must be treated separately.

\medskip

We consider first the case where $\epsilon=1$.
From the  Schur  superpolynomial labeled by a single-part superpartition, either fermionic or bosonic, we can prove inductively that
\be\label{lemreltypeIabc}
B_n^{(1)}s_\La = s_{\Om} \qquad \text{if} \qquad \Om^* = (n, \La^*) \quad \text{and}\quad \Om^\circledast=(n+1, \La^\circledast)
\ee
using a combinatorial argument constructed from the same involution as in Section \ref{Sproof_typeIe}.  

\begin{proof}
We keep the same representation as in Section \ref{Sproof_typeIe} for the diagrams that appear in the expansion of $B_n^{(1)}s_\La$,
\be\label{l_expan_BunmsurS}
B_n^{(1)}s_\La =  \sum_{r\geq 0} (-1)^r \ta_{n+r+1}^{\phantom \perp} \circ e_r^\perp \, s_\La.
\ee
In particular, boxes that are removed from $\La$ using the inverse of the Pieri rule of type $e_rs^*_\La$ form a vertical $r$-strip and are marked with the symbol $\times$.  
The boxes and the circle that are then added to the resulting diagrams using the Pieri rule of type $\ta_{n+r+1} s_\La$ form a fermionic horizontal $(n+r)$-strip, and are marked with the symbol $+$.  In addition, we mark with the symbol ``!'' a circle that has been displaced  from its original row (a circle that is moved only horizontally, that is, which remains in the same row, is not marked with  a ``!'').  A box that has been removed in the first step and then added in the second step, is marked with the symbol $\PlusX$.  Here is an example of this representation:
\be\label{ex1_pour_B_proof}
\ta_{6}^{\phantom \perp}\circ e_1^\perp \cdot
	\scalebox{1.1}{$\superY{
	\,&\,&\,& \\
	\,&\,&\yF{} \\
	\,& \\
	\yF{}
	}$}
\supset
	\scalebox{1.1}{$\superY{
	\,&\,&\,&\PlusX&+&+& \yF{+}  \\
	\,&\,&+&\yF{} \\
	\,&\, \\
	+ \\
	\yF{!}
	}$}
\ee

Now, for each (decorated) diagram that appears in  the expansion \eqref{l_expan_BunmsurS}, we identify a  \emph{transmutable} box (there is at most one such box), which is defined as follows:  reading  the diagram  from top right to bottom left, a transmutable box is the first box that satisfies the following conditions:
\begin{itemize}
\item[1.]  it is an empty box or a $\PlusX$box;
\item[2.] it is a box that belongs to the set of removable boxes (of the original diagram);
\item[3.] it does not lie above a $+$box or above an empty box;
\item[4.] there is no unmarked (empty) circle to its right.
\end{itemize}
Note that this definition of transmutable box is the same as before even though the Bernstein operators act differently.    
An empty transmutable box is marked with a $\bullet$.  There is a unique diagram in the expansion of the terms stemming from  \eqref{l_expan_BunmsurS} which contains no transmutable box. 
  It comes from the $r=0$ sector, where one $+$box is added to every column, and the added circle is in the first row.  
  Since $n \geq \La_1^{\circledast}$, this diagram always exists.  In our example \eqref{ex1_pour_B_proof}, such a diagram corresponds to
\be\label{ex2_pour_B_proof}
\ta_{5}^{\phantom \perp}\circ e_0^\perp \cdot
	\scalebox{1.1}{$\superY{
	\,&\,&\,& \\
	\,&\,&\yF{} \\
	\,& \\
	\yF{}
	}$}
\supset
	\scalebox{1.1}{$\superY{
	\,&\,&\,& \,& \yF{+}  \\
	\,&\,&+&+ \\
	\,&\, & \yF{!}\\
	+ &+\\
	\yF{!}
	}$}
\ee

The involution $\mathcal{J}$ (which is the same involution as in Section \ref{Sproof_typeIe}) interchanges an empty transmutable box (marked with a $\bullet$) with a transmutable $\PlusX$box (marked with a $\PlusXb$), and vice-versa.  This involution thus pairs all diagrams but one in \eqref{l_expan_BunmsurS} with an identical diagram having opposite sign.  The left-over digram  is the unique digram containing no transmutable box (described above); it is the one associated with the superpartition
\be
\Om = (n, \La_1^{\rm a}, \ldots, \La_{\rm m}^{\rm a}; \La_{\rm{m+1}}^{\rm s}, \ldots)
\ee
as desired.  This ends the proof of \eqref{lemreltypeIabc}.\\
\end{proof}

Finally, there remains to treat the case where $\epsilon=0$, that is, to show that
\be
B_n^{(0)}\,s_\La = s_\Om \qquad \text{if} \qquad \Om^*=(n,\La^*) \quad \text{and} \quad \Om^\circledast = (n, \La^\circledast).
\ee
\begin{proof}
Let $\La$ be a superpartition and let $n\geq \La_1^\circledast$.  Using the defining relation \eqref{B0vsB1}, we can write
\be
B_n^{(0)} \, s_\La = \partial_{\tilde{e}_0} \circ B_n^{(1)}\, s_\La =  \partial_{\tilde{e}_0} s_\Delta
\ee
where $\Delta^*=(n,\La^*), \Delta^\circledast=(n+1, \La^\circledast)$, and where the last step follows from \eqref{lemreltypeIabc}.
  We thus have to show that the action of $\partial_{\tilde{e}_0}$ on a Schur superfunction associated to a diagram starting with a fermionic part (i.e.~the first row contains a circle) produces one Schur function with the same diagram except that the circle  of the first row has been striped.  This follows from the same argument used to prove  \eqref{lemtypeIefermion}
 and Lemma \ref{Lemlemlemdeleo}.  
\end{proof}

\end{appendix}

\bibliography{ref1}

\begin{thebibliography}{10}
\providecommand{\url}[1]{#1}
\csname url@samestyle\endcsname
\providecommand{\newblock}{\relax}
\providecommand{\bibinfo}[2]{#2}
\providecommand{\BIBentrySTDinterwordspacing}{\spaceskip=0pt\relax}
\providecommand{\BIBentryALTinterwordstretchfactor}{4}
\providecommand{\BIBentryALTinterwordspacing}{\spaceskip=\fontdimen2\font plus
\BIBentryALTinterwordstretchfactor\fontdimen3\font minus
  \fontdimen4\font\relax}
\providecommand{\BIBforeignlanguage}[2]{{%
\expandafter\ifx\csname l@#1\endcsname\relax
\typeout{** WARNING: IEEEtranS.bst: No hyphenation pattern has been}%
\typeout{** loaded for the language `#1'. Using the pattern for}%
\typeout{** the default language instead.}%
\else
\language=\csname l@#1\endcsname
\fi
#2}}
\providecommand{\BIBdecl}{\relax}
\BIBdecl

\bibitem{BDLM12_2}
O.~Blondeau-Fournier, P.~Desrosiers, L.~Lapointe, and P.~Mathieu, ``Macdonald
  polynomials in superspace as eigenfunctions of commuting operators,''
  \emph{J. of Comb.}, vol.~3, no.~3, pp. 495--561, 2012,
  \textsf{arXiv:1202.3922 [\mbox{math-ph}]}.

\bibitem{BDLM12}
------, ``Macdonald polynomials in superspace: conjectural definition and
  positivity conjectures,'' \emph{Lett. Math. Phys.}, vol. 101, no.~1, pp.
  27--47, 2012, \textsf{arXiv:1112.5188 [\mbox{math-ph}]}.

\bibitem{BDM15}
O.~Blondeau-Fournier, P.~Desrosiers, and P.~Mathieu, ``Supersymmetric
  {R}uijsenaars-{S}chneider model,'' \emph{Phys. Rev. Lett.}, vol. 114, no.~12,
  p. 121602, 2015, \textsf{arXiv:1403.4667 [\mbox{hep-th}]}.

\bibitem{BLM15}
O.~Blondeau-Fournier, L.~Lapointe, and P.~Mathieu, ``Double {M}acdonald
  polynomials as the stable limit of {M}acdonald superpolynomials,'' \emph{J.
  Algebraic Combin.}, vol.~41, no.~2, pp. 397--459, 2015,
  \textsf{arXiv:1211.3186 [\mbox{math-ph}]}.

\bibitem{BM1}
O.~Blondeau-Fournier and P.~Mathieu, ``Schur superpolynomials: combinatorial
  definition and {P}ieri rule,'' \emph{SIGMA}, vol.~11, 2015,
  \textsf{arXiv:1408.2807 [\mbox{math-ph}]}.

\bibitem{Goulden}
S.~R. Carrell and I.~P. Goulden, ``Symmetric functions, codes of partitions and
  the {K}{P} hierarchy,'' \emph{J. Algebraic Combin.}, vol.~32, pp. 211--226,
  2010, \textsf{arXiv:0902.4441 [\mbox{math.CO}]}.

\bibitem{DLM_CMS}
P.~Desrosiers, L.~Lapointe, and P.~Mathieu, ``Supersymmetric
  {C}alogero--{M}oser--{S}utherland models and {J}ack superpolynomials,''
  \emph{Nucl. Phys.}, vol. B606, no.~3, pp. 547--582, 2001,
  \textsf{arXiv:0210190 [\mbox{hep-th}]}.

\bibitem{DLM_basis}
------, ``Classical symmetric functions in superspace,'' \emph{J. Algebraic
  Combin.}, vol.~24, pp. 209--238, 2006, \textsf{arXiv:0509408
  [\mbox{math.CO}]}.

\bibitem{Jarvis}
P.~D. Jarvis and C.~M. Yung, ``Symmetric functions and the {K}{P} and {B}{K}{P}
  hierarchies,'' \emph{J. Phys. A}, vol.~26, p. 5905, 1993.

\bibitem{JL17}
M.~Jones and L.~Lapointe, ``Pieri rules for {S}chur functions in superspace,''
  \emph{J. of Comb. Theor.}, vol. 148, pp. 57--115, 2017,
  \textsf{arXiv:1608.08577 [\mbox{math.CO}]}.

\bibitem{KvdL_tB}
V.~G. Kac and W.~V. de~Leur, ``Super boson-fermion correspondence of type
  {B},'' \emph{Infinite-Dimensional Lie Algebras and Groups (Luminy-Marseille,
  1988), Adv. Ser. Math. Phys}, vol.~7, pp. 369--406, 1989.

\bibitem{KvdL}
V.~G. Kac and W.~V.~D. Leur, ``Super boson-fermion correspondence,'' \emph{Ann.
  Inst. Fourier}, vol.~37, pp. 99--137, 1987.

\bibitem{MacSym95}
I.~Macdonald, \emph{Symmetric functions and {Hall} polynomials}, 2nd~ed., ser.
  Oxford Mathematical Monographs.\hskip 1em plus 0.5em minus 0.4em\relax
  Oxford: Clarendon Press, 1995.

\bibitem{Miwa00}
T.~Miwa, M.~Jimbo, and E.~Date, \emph{Solitons: Differential equations,
  symmetries and infinite dimensional algebras}.\hskip 1em plus 0.5em minus
  0.4em\relax Cambridge University Press, 2000.

\end{thebibliography}
\bibliographystyle{IEEEtranS}
\end{document}